\documentclass[letterpaper]{article} 
\usepackage{aaai2026}  
\usepackage{times}  
\usepackage{helvet}  
\usepackage{courier}  
\usepackage[hyphens]{url}  
\usepackage{graphicx} 
\usepackage{natbib}  
\usepackage{caption} 
\usepackage{chrkroer}

\usepackage{algorithm}
\usepackage{algorithmicx, algpseudocode}

\usepackage{newfloat}
\usepackage{listings}
\DeclareCaptionStyle{ruled}{labelfont=normalfont,labelsep=colon,strut=off} 
\floatstyle{ruled}
\newfloat{listing}{tb}{lst}{}
\floatname{listing}{Listing}
\newcommand{\discsoldiers}{k}
\newcommand{\contsoldiers}{\sigma}
\newcommand{\cmark}{\checkmark}
\newcommand{\xmark}{\ding{55}} 
\usepackage{multirow}
\usepackage{booktabs}
\usepackage{amssymb}
\usepackage{pifont} 
\usepackage{cleveref}
\usepackage{thmtools, thm-restate}
\usepackage{subcaption}
\usepackage{tikz}

\title{Colonel Blotto with Battlefield Games}
\author {
    Salam Afiouni\textsuperscript{\rm 1},
    Jakub \v{C}ern\'{y}\textsuperscript{\rm 1},
    Chun Kai Ling\textsuperscript{\rm 2},
    Christian Kroer\textsuperscript{\rm 1}
}
\affiliations {
    \textsuperscript{\rm 1}Columbia University, USA\\
    \textsuperscript{\rm 2}National University of Singapore, Singapore\\
    sa4316@columbia.edu, 
    jakub.cerny@columbia.edu,
    chunkail@nus.edu.sg, 
    christian.kroer@columbia.edu
}

\begin{document}

\maketitle

\begin{abstract}
We study a class of two-player zero-sum Colonel Blotto games in which, after allocating soldiers across battlefields, players engage in (possibly distinct) normal-form games on each battlefield. Per-battlefield payoffs are parameterized by the soldier allocations. This generalizes the classical Blotto setting, where outcomes depend only on relative soldier allocations. We consider both discrete and continuous allocation models and examine two types of aggregate objectives: linear aggregation and worst-case battlefield value. For each setting, we analyze the existence and computability of Nash equilibrium. The general problem is not convex-concave, which limits the applicability of standard convex optimization techniques. However, we show that in several settings it is possible to reformulate the strategy space in a way where convex-concave structure is recovered. We evaluate the proposed methods on synthetic and real-world instances inspired by security applications, suggesting that our approaches scale well in practice.
\end{abstract}


\section{Introduction}
Colonel Blotto games are a classic game class used to model competitive resource allocation \cite{borel1991application,roberson2006colonel,kovenock2012conflicts}. Introduced by Émile Borel in the 1920s, Blotto games describe situations where players simultaneously allocate soldiers (resources) between battlefields, after which each battlefield is ``won'' by the player who allocated more soldiers to it. In the classic setting, each battlefield features a winner-takes-all property, creating an interesting game theoretic conundrum where both players would like to either just barely win or badly lose every battlefield. Blotto games and their extensions have been studied in many disciplines including computer and social sciences, with applications to politics \citep{che2008caps,laslier2002distributive,myerson1993incentives} , warfare \citep{gross1950continuous,shubik1981systems}, and behavioral science \citep{chowdhury2013experimental,arad2012multi}. 

In this paper, we consider a variant of a two-player Colonel Blotto game played over $n$ heterogeneous battlefields that features two levels of play. In our setting, each battlefield $i\in[n]$ also comprises a zero-sum ``subgame'' $G_i$, the structure and payoffs of which depends on the soldiers each player allocated to $i$. This captures many real-world settings where players follow-up with individual battlefield-level strategies. For instance, a political party first allocates money between state elections, after which, each individual state decides how they should spend that money. Likewise in warfare, the tactics at a battlefield level play a huge part in the success of the overarching conflict; these battlefield level tactics are themselves games played at a lower level, the outcomes of which depend greatly on the number of soldiers assigned.
In more abstract terms, players first simultaneously determine a soldier assignment over battlefields (unobserved by the opponent), then play, in \textit{every} battlefield, a subgame whose payoff matrix is parametrized by the player's soldier allocation on that battlefield. We refer to this game as the ``two-level Blotto game''. By convention, Player 1 (resp. Player 2) seeks to minimize (resp. maximize) the payoff. Accordingly, we often refer to Player 1 (resp. Player 2) as the minimizing (resp. maximizing) player.

We analyze three dimensions of the proposed two-level Blotto games. First, we study \textit{discrete versus continuous} soldier types; the former implies that individual soldiers cannot be subdivided, though players are allowed to adopt randomization. Second, we study two ways of \textit{aggregating} payoffs from individual battlefields -- either with sum or min aggregators. Both these axes have been extensively explored in the Blotto literature (see, e.g,~\citet{doi:https://doi.org/10.1002/9780470400531.eorms1022,doi:https://doi.org/10.1002/9780470400531.eorms0913}, or~\citet{vu2020models}).
Third, we study two \textit{strategic settings}: the two-sided case, where both players allocate soldiers across battlefields and the simpler one-sided case, in which only the maximizing player is allowed soldier allocations. The one-sided setting is motivated by security applications, where only the ``defending'' player has soldiers to allocate, and provides stronger guarantees of equilibrium existence and computational tractability.

Our contributions are as follows. 
First, for every combination of soldier type (discrete vs. continuous), aggregator (sum vs. min), and strategic setting (two-sided vs. one-sided), we address the following: (i) does a Nash equilibrium exist, and (ii) are the max-min and min-max strategies well-defined and attained?
\footnote{A positive answer to both questions implies a valid minimax theorem in that regime. Moreover, any existence result in the two-sided model implies the one-sided case by fixing the minimizer’s allocations to zero.}
Our results are summarized in~\Cref{tab:results}. Second, for each of the cases where an equilibrium exists, we provide algorithms for solving them. In most cases, these problems reduce to linear programs, though in one special case the problem reduces to a quasiconcave problem. Third, we empirically evaluate these algorithms on synthetic data, demonstrating that these algorithms are usable.

\paragraph{Related work.} Due to space constraints, we only provide a small sampling of relevant work. Computationally, \citet{ahmadinejad2019duels} provided an LP based approach for separable battlefields (similar to one of our settings) to solve Blotto games but using the ellipsoid method; their method extends to continuous \textit{Lotto} games~\cite{hart2008discrete,dziubinski2013non}, which admit infinite action spaces, though they do not have subgames. \citet{behnezhad2023fast} avoid the impractical ellipsoid algorithm and propose a practical LP formulation based on layered graphs, and is very closely related to one of our approaches. \citet{kvasov2007contests} and \citet{roberson2012non} study a non-constant (zero) sum variant of Blotto games, while \citet{hortala2012pure} study pure strategy equilibria in such variants. \citet{kovenock2012coalitional} study a team variant where a single player plays against a coalition of 2 other players while \citet{boix2020multiplayer} study more general multiplayer settings. \citet{stephenson2024multi} study payoffs which depend on a function of allocated soldiers per player. \citet{powell2009sequential} and \citet{hausken2012impossibility} consider sequential variants of Blotto games.

\begin{table}[t]
\centering
\begin{tabular}{@{}llcc@{}}
\toprule
\textbf{Setting} &  & \textbf{NE} & \textbf{Mx/Mn exist.} \\
\midrule
\multirow{3}{*}{Discrete} 
  & Two-sided sum    & \cmark & \cmark  \\
  & Two-sided min    & \xmark & \cmark  \\
  & One-sided min    & \cmark & \cmark  \\
\midrule
\multirow{6}{*}{Continuous} 
  & Two-sided sum     & \xmark & $\;\checkmark^*$ \\
  & One-sided sum     & \xmark & $\;\checkmark^*$ \\
  & One-sided sum (L) & \cmark & \cmark          \\
  & Two-sided min     & \xmark & $\;\checkmark^*$ \\
  & One-sided min     & $\;\checkmark^*$ & $\;\checkmark^*$ \\
  & One-sided min (L) & \cmark & \cmark          \\
\bottomrule
\end{tabular}
\caption{Existence of NE and max-min strategies in different settings of discrete and continuous Blotto games. $^*$ holds only when the utility \(u\) is continuous in the soldiers allocation. (L) means that battlefield utilities are linear in the maximizing player's soldier allocation.}
\label{tab:results}
\end{table}
\section{Preliminaries}
Let $n > 0$ be the number of battlefields, and $m^{1}, m^{2}$ be the total number of soldiers available to Players 1 and 2. 
Let $\mathcal Z^j$ denote the set of all possible soldier assignments $z^j$ for player $j\in\{1,2\}$, where $z^j=(z^j_1,\dots,z^j_n)$ is a vector specifying soldier distributions across the $n$ battlefields, subject to the budget constraint $\sum_{i=1}^n z_i^j=m^j$. 
When soldiers are discrete, we let $z^j=k^j\in \mathbb Z^n_{\geq 0}$, and when soldiers are continuous, $z^j=\sigma^j\in \mathbb R^n_{\geq 0}$. We denote by $\mathcal K^j$ and $\Sigma^j$ the sets of all possible soldier assignments for player $j$ in the discrete and continuous case, respectively. 
Each subgame $G_i$ is a finite normal-form zero-sum game defined as a tuple $(\mathcal A_i^1,\mathcal A_i^2,u_i)$, where $\mathcal A_i^j$ denotes the action space of player $j\in\{1,2\}$ in battlefield $i$, and $u_i:\mathcal A_i^1\times \mathcal A_i^2\times \mathcal Z^1\times \mathcal Z^2\rightarrow \mathbb R$ denotes the payoff in that battlefield that Player 2 (resp. Player 1) seeks to maximize (resp. minimize). For each battlefield $i\in[n]$, the payoff $u_i$ is a function of both the number of soldiers players allocated to that battlefield and the actions they play in the subgame $G_i$.  We denote by $\alpha_i^{j}\in \mathcal{A}_i^{j} $ the action player $j$ plays on battlefield $i$.
For battlefield $i$ and any allocation profile $(z_i^1,z_i^2)$, we let $v_i^*(z_i^1,z_i^2)$ be the equilibrium value of subgame $G_i$ when player $j$ allocates $z_i^j$ soldiers to it. 
To compute the overall utility in the two-level Blotto game, we consider two aggregation functions: 1) the sum of the expected utilities across battlefields $U= \sum_{i\in [n]} \mathbb E[u_i(.)]$, and 2) the minimum of the expected utilities in all battlefields $U=\min_{i\in[n]} \mathbb E[u_i(.)]$. 

Across the variants of the two-level Blotto game we analyze, players choose mixed strategies over their subgame actions. 
The use of randomization at the allocation stage depends on whether the model is discrete or continuous. In the discrete case, the allocating players (the maximizing player in the one-sided case and both players in the two-sided case) randomize over integer allocation vectors. In the continuous case, players could also mix over allocations; yet we show that in some settings, equilibrium can be guaranteed even when players commit to a deterministic (pure) allocation. This motivates our focus on pure strategies at the allocation stage, which offer simpler representation and tractable computation.
The various strategy representations are summarized in~\Cref{app:notation}. Finally, a strategy pair $(x^*,y^*)\in X\times Y$ is a Nash equilibrium (NE) if $U(x^*,y)\leq U(x^*,y^*)\leq U(x,y^*)$ for all $x\in X,y\in Y$, where $X$ and $Y$ are the strategy spaces of Player 1 and 2, respectively.

Minimax theorems play a central role in our analysis, particularly in establishing equilibrium existence. In fact, the existence of a minimax theorem for a setting is equivalent to equilibrium existence (\Cref{prop:nash_minimax}). We rely on Sion’s minimax theorem~\citep{pjm/1103040253} in several of our proofs and also examine its applicability across various settings. In the setting where only Player 2 allocates a continuum of soldiers across battlefields and individual subgame payoffs are aggregated via the minimum operator, we apply a variant known as the Kneser-Fan minimax theorem~\citep{pjm/1103040253}. 
\begin{restatable}{proposition}{nashminimax}
\label{prop:nash_minimax}
    Let $u:X\times Y\to \mathbb{R}$ be an arbitrary function on arbitrary sets $X,Y$. Then $\max_{x\in X}\inf_{y\in Y}u(x,y) = \min_{y\in Y}\sup_{x\in X} u(x,y)$ if and only if there exists a NE in a two-player zero-sum game with action spaces $X$ and $Y$ and utility function $u$.
\end{restatable}
 \begin{theorem}[Sion's minimax theorem] Let $X$ and $Y$ be nonempty convex and compact subsets of two linear topological spaces, and $f: X \times Y \rightarrow \mathbb R$ be a function that is upper semicontinuous and quasiconcave in the first variable and lower semicontinuous and quasiconvex in the second. Then \(
 \min_{y \in Y} \max_{x\in X} f(x,y) = \max_{x\in X} \min_{y \in Y} f(x,y)
 \).
 \end{theorem}
\begin{definition}[Concavelike, Convexlike, and Concave-convexlike]\label{def:ccv-cvx-like}
Consider two spaces $X$ and $Y$, and let $f$ be a function $f$ on $X\times Y$. We say $f$ is \textit{concavelike} in $X$ if for every $x_1,x_2\in X$ and $0\leq t \le 1$, there is an $x \in X$ such that
    \(
    tf(x_1,y)+(1-t)f(x_2,y) \leq f(x,y) \; \forall y \in Y
    \), and $f$ is \textit{convexlike} in $Y$ if for every $y_1,y_2\in Y$ and $0\leq t \le 1$, there is a $y \in Y$ such that
    \(
    tf(x,y_1)+(1-t)f(x,y_2) \geq f(x,y) \; \forall x \in X
    \). 
  Finally, $f$ is called concave-convexlike if it is concavelike in $X$ and convexlike in $Y$.
\end{definition}
\begin{theorem}[Kneser-Fan minimax theorem]\label{thm:kneser-fan}
    Let $X$ be compact, $Y$ any space, and $f$ a function on $X \times Y$ that is concave-convexlike. If $f(x, y)$ is upper semicontinuous in $x$ for each $y$, then \(\sup_{x\in X} \inf_{y\in Y} f =\inf_{y\in Y} \sup_{x\in X} f. \)
\end{theorem}
\section{Colonel Blotto with Discrete Soldiers}
In the \textit{discrete two-level Blotto} game, players allocate a finite number of indivisible soldiers across battlefields and randomize over both soldier assignments and subgame strategies. Let $\discsoldiers^j=(\discsoldiers^j_1,\cdots,\discsoldiers^j_n)\in \mathcal K^j$ be player $j$’s discrete allocation vector. Then the payoff in battlefield $i$ is given by $u_i(\alpha^{1}_i, \alpha^{2}_i, \discsoldiers_i^{1}, \discsoldiers_i^{2})$.
In our two-level formulation, each player 
$j$ adopts a two-level strategy composed of 1) a probability distribution $\gamma^j$ over soldier assignments, and 2) for each battlefield $i$ and each soldier count $k_i^j$, a mixed strategy $\delta^j_{i,k_i^j}$ over actions $\alpha_i^j\in \mathcal A_i^j$ in subgame $G_i$. The action set in each battlefield is independent of the number of soldiers allocated to that battlefield (this may be relaxed easily). Since the number of possible assignments is finite, each subgame can be viewed as a Bayesian game in which player $j$ has types $\{0,\dots,m^{j}\}$, each occurring with probability equal to the chance that the player assigns that number of soldiers to the battlefield under $\gamma^j$.  

\subsection{Sum Aggregator}
\label{sec:sum-agg}
The overall utility under the sum aggregator is given by 
\begin{multline*}
U(\delta,\gamma)
= \sum_{i\in[n]}\sum_{k^1\in\mathcal K^1}\sum_{k^2\in \mathcal K^2}
  \gamma^1(k^1)\,\gamma^2(k^2)\\
  \times
  \sum_{\alpha_i^1\in \mathcal A_i^1}\sum_{\alpha_i^2\in \mathcal A_i^2}
  \delta^1_{i,k_i^1}(\alpha_i^1)\,\delta^2_{i,k_i^2}(\alpha_i^2)\,
  u_i(\alpha_i^1,\alpha_i^2,k_i^1,k_i^2)\,.
\end{multline*}
In this setting, we can formulate the two-level Blotto game as a linear program. For each battlefield $i$, and given a fixed distribution over $k^{1}_i$ and $k^{2}_i$, the optimal subgame strategy corresponds to solving a Bayesian extensive-form game where nature randomizes over the $(m^{1}+1)\cdot (m^{2}+1)$ possible types. We model this randomization in two sequential phases -- first determining Player 1’s probabilities, then Player 2’s -- thereby inducing a product distribution over the joint types. 
By translating the type assignment into two stages, we can assign the first phase to Player 1 and the second phase to Player 2 while retaining perfect recall. 

More precisely, for player $j\in\{1,2\}$, we define a flow polytope $\Gamma^j$ with two types of variables. First, the variable $x^{j}_{i,k_i^j,\alpha_i^j}$ denotes the probability that player $j$ allocates $k_i^j$ soldiers to battlefield $i$ and plays action $\alpha_i^j$ in subgame $G_i$. Second, the flow variable $h^{j}_{i,a,b}$ denotes the flow through battlefield $i$ for player $j$: it represents the transition from $a$ remaining soldiers to $b$ remaining soldiers after allocating $a-b$ soldiers to battlefield $i$. The polytope $\Gamma^j$ is defined by the following constraints:
\begin{flalign*}\label{eq:Gamma-def}
\Gamma^j = \biggl\{&h^{j}_{i,a,b},
x^{j}_{i,k_i^j,\alpha_i^j} \in [0,1]
\text{ such that :} && \\
&h^j_{0,m^j,m^j}=1, \quad h^j_{n-1,0,0}=1, \\
&h^j_{0,a,b}=0 \,\, \forall a,b\in[0,m^j],\, (a,b)\ne(m^j,m^j), \\
&\sum_{a=c}^{m^j}h^j_{i-1,a,c}
        = \sum_{b=0}^{c}h^j_{i,c,b}
        \,\, \forall\,i\in[1,n-1],\,c\in[0,m^j],\\
&\sum_{\alpha_i^j\in\mathcal A^j_i}x^j_{i,k_i^j,\alpha_i^j}
        = \sum_{k_i^j \leq r \leq m^j} h^{j}_{i, r, r -k_i^{j}} 
        \,\, \forall i\in [1,n-1],\\ &k_i^j\in [0,m^{j}]\biggr\}.
\end{flalign*}
By encoding strategies with a flow polytope, we obtain a polynomial‑size representation of strategies that is payoff equivalent to the original strategy space $\Delta(\mathcal K^j)\times \prod_{i\in [n]}\Delta(\mathcal A_i^j)$. This flow polytope is a generalization of the layered graph approach of \citet{behnezhad2023fast}, with the addition of subgames.

\begin{restatable}[Kuhn's theorem under two-sided sum aggregator]{theorem}{kuhntwosidedsumdisc}
\label{thm:kuhn-2sided-sum-disc}
    Consider a discrete two-level Blotto game with sum aggregator, and suppose $\gamma^j\in\Delta(\mathcal K^j)$, $\delta_i^j \in \Delta(\mathcal A_i^j)$, $j\in\{1,2\}$. Then the overall payoff is bilinear in the strategies $x^{1}$ and $x^{2}$. Specifically, 
    \(
    U = \sum_{i} \sum_{\alpha_i^{1}} \sum_{\alpha_i^{2}}
    \sum_{k_i^{1}}
    \sum_{k_i^{2}}
    u_i (\alpha_i^{1}, \alpha_i^{2}, k^{1}_i, k^{2}_i) \cdot x^{1}_{i,k_i^{1},\alpha_i^{1}}
    \cdot x^{2}_{i,k_i^{2},\alpha_i^{2}}.
    \)
\end{restatable}

\Cref{thm:kuhn-2sided-sum-disc} implies that rather than optimizing over $\Delta(\mathcal K^j)\times \prod_{i\in [n]}\Delta(\mathcal A_i^j)$, we can optimize over a flow polytope $\Gamma^j$. Computing a NE may be expressed as a bilinear saddle point problem $\max_{x^2\in\Gamma^2} \min_{x^1\in \Gamma^1} U$.
Since $\Gamma^1$ and $\Gamma^2$ are compact and convex sets, Sion's minimax theorem holds. Then, by~\Cref{prop:nash_minimax}, we have:
\begin{corollary}\label{cor:two-sided-sum-disc}
     Under the sum-aggregator, the discrete two-level Blotto game admits a NE.
\end{corollary}

We propose two methods to solve the bilinear saddle point problem. The first uses a common trick of taking the dual of the inner optimization problem, converting the min-max problem into a single minimization problem that can be solved by a linear program \cite{behnezhad2023fast}. For brevity, details of the LP are deferred to Appendix~\ref{dual:two-sided-disc-sum}. The second method is based on online learning and self-play. In the classic setting of zero-sum matrix games, each player's strategy space is the probability simplex, and it is known that the recommendations given by the online learners converge on average to a NE \citep{hart2000simple,roughgarden2010algorithmic}. This is readily adapted to imperfect information extensive-form games by defining regret minimizers over the \textit{treeplex} 
 \citep{zinkevich2007regret}. In our setting, the regret minimizers are over the flow polytopes $\Gamma^1$ and $\Gamma^2$, which may be done using scaled extensions \citep{farina2019efficient} or kernel-based approaches \citep{farina2022kernelized,takimoto2003path}. For brevity, we defer details to Appendix~\ref{sec:online-learning-appendix-discrete-sum-agg}, but remark that the regret minimizers require $\mathcal{O}(n\cdot \max_{i,j} \{ m^j (m^j + |\mathcal{A}^j_i|)\})$ space (i.e., \textit{linear} in the size of the flow polytope), $\mathcal{O}(n\cdot \max_{i,j} \{ m^j (m^j + |\mathcal{A}^j_i|)\} + \max_{i} \{ m^1_i m^2_i |\mathcal{A}^1_i| |\mathcal{A}^2_i| \} )$ time-per-iteration  and incurs total regret at a rate of $\mathcal{O}(\sqrt{T})$, depending on the implementation details. 
 Our methods readily extend to settings where each subgame is a perfect recall extensive-form game. This is done by reformulating the strategy space in each $G_i$ by the treeplex \citep{von1996efficient}. 
 Our results here extend trivially to the one‑sided case. %

\begin{restatable}{proposition}{comptwosidedsumdisc} \label{prop:comp-2sided-sum-disc}
Under the sum aggregator, a NE of the discrete two-level Blotto game is polynomial-time computable.
\end{restatable}

\begin{remark}
    Note that all results extend directly to the one-sided setting, in which the minimizing player does not participate in the soldier allocation.
\end{remark}

\subsection{Min Aggregator}\label{sect:disc-min}
When the overall utility is defined as the minimum across all battlefields, the aggregate utility is given by
\begin{multline*}
    U(\delta,\gamma) = \min_{i\in[n]}\sum_{k^1\in\mathcal K^1}\sum_{k^2\in\mathcal K^2} \gamma^1(k^1)  \gamma^2(k^2) \\
    \times \biggl(\sum_{\alpha_i^1\in \mathcal A_i^1} \sum_{\alpha_i^2\in \mathcal A_i^2} \delta_{i,k_i^1}^1(\alpha_i^1)\delta^2_{i,k_i^2}(\alpha_i^2) u_i(\alpha_i^1,\alpha_i^2,k_i^1,k_i^2) \biggr).
\end{multline*}
Unlike with the sum aggregator, when both players allocate soldiers to battlefields, the two-level Blotto game may not admit a NE even when battlefield utilities are continuous.

\begin{restatable}{theorem}{existtwosidedmindisc}\label{thm:non-exist-2sided-min-disc} Consider a two-sided discrete two-level Blotto game under the min aggregator. Even if all battlefield utilities are continuous, a NE may not exist. Nevertheless, both the max-min and min-max values are well‐defined and are attained.
\end{restatable}

This motivates our study of the one-sided model, in which only the maximizing player allocates soldiers to battlefields, then both players engage in independent zero-sum subgames. This variant captures common security scenarios where for example, only the defender allocates resources, and guarantees equilibrium existence while admitting efficient computation.
In this setting, the total utility is 
\begin{multline*}
    U(\delta,\gamma)= \min_{i\in[n]}\sum_{k^2\in\mathcal K^2} \gamma^2(k^2) \\
    \times \biggl(\sum_{\alpha_i^1\in \mathcal A_i^1} \sum_{\alpha_i^2\in \mathcal A_i^2}\delta_i^1(\alpha_i^1)\delta^2_{i,k_i^2}(\alpha_i^2) u_i(\alpha_i^1,\alpha_i^2,k_i^2) \biggr).
\end{multline*}
This formulation admits a tractable NE.  
\begin{restatable}[Kuhn's theorem under one-sided min aggregator]{theorem}{kuhnonesidedmindisc}\label{thm:kuhn-one-sided-disc-min}
    Consider a one-sided discrete two-level Blotto game with min aggregator, and
    suppose $\gamma^2\in\Delta(\mathcal K^2)$, $\delta_i^j\in \Delta(\mathcal A_i^j),j\in\{1,2\}$.
    Then the overall payoff is convex and bilinear in $p$ and $x^{2}$, where $p\in \Delta(n)$ is a convexification variable and $x^2$ is the treeplex strategy for Player 2. Specifically, 
    \( 
    U =
    \min_{p} \sum_{i} \sum_{\alpha_i^2} \sum_{k_i^2} p_i \, \mathbb E_{\alpha_i^1 \sim \delta^1_i} [u_i(\alpha_i^1,\alpha_i^2,k_i^2)] x^2_{i,k_i^2,\alpha_i^2}.
    \)
\end{restatable}

Using the bilinear formulation from~\Cref{thm:kuhn-one-sided-disc-min} we can now express the min-max problem as 
\(
\min_{\delta^{1}}\max_{x^2}\min_{p}  \sum_{i}\sum_{\alpha_i^2}\sum_{k_i^2}p_i x^2_{i,k_i^2,\alpha_i^2} \mathbb E[u_i(\alpha_i^{1},\alpha_i^2,k_i^2)].
\)
Exploiting the bilinear structure of the payoff, we invoke Sion’s minimax theorem to interchange the inner maximization and minimization operators. By then combining the two minimizations, we define a sequence-form strategy over a two-level game for Player 1, in which he first chooses a battlefield and then plays an action in it, using $y^1(i,\alpha^1_i)$ to denote the product $p_i\cdot\delta^1_i(\alpha^1_i).$ Under this reformulation, we recover the game's max-min structure, which yields the following result.

\begin{restatable}{theorem}{siononesidedmindisc}\label{thm:sion-one-sided-disc-min}
    Consider a one-sided discrete two-level Blotto game with min aggregator. Then the minimax theorem holds. Specifically, \(
\min_{\delta^{1}}\max_{\delta^{2},\gamma^2}\min_{i} \mathbb E[u_i(\alpha_i^{1},\alpha^{2}_i,k_i^2)] =  \max_{\delta^{2},\gamma^2}\min_{\delta^{1}}\min_{i} \mathbb E[u_i(\alpha_i^{1},\alpha_i^{2},k_i^2)]
    \).
\end{restatable}

The flow polytope representation of Player 2's strategy introduces a set of constraints captured by the polytope $\Gamma^2$. Analogously, the sequence-form product representation of Player 1's strategy induces the constraints defined over the following polytope 
\[
\mathcal{P} = \left\{
\begin{array}{ll}
y^1(i) \geq 0 & \\
y^1(i, \alpha_i^1) \geq 0 & \\
\end{array} \middle|\
\begin{array}{ll}
    y^1(\emptyset) = 1, \\ 
    \sum\limits_{i \in [n]} y^1(i) = y^1(\emptyset), \\ 
    \sum\limits_{\alpha_i^1 \in \mathcal{A}_i^1} y^1(i, \alpha_i^1) = y^1(i). \\ 
\end{array}
\right\}
\]
Therefore, a NE can be computed by solving a LP over the polytopes $\Gamma^2$ and $\mathcal P$, obtained by dualizing either player’s best-response program (see \Cref{dual:one-sided-disc-min}). Since both $\Gamma^2$ and $\mathcal P$ have polynomial-size representations, we have:
 
\begin{restatable}{corollary}{componesidedmindisc}\label{cor:comp-one-sided-min-disc}
In the one-sided model, a NE of the two-level Blotto game with min aggregator exists and can be computed in polynomial time.
\end{restatable}
\section{Colonel Blotto with Continuous Soldiers}
In the \textit{continuous two-level Blotto game}, we focus on equilibria where players choose deterministic (pure) soldier allocations and may randomize over subgame strategies. 
Unlike the discrete setting where equilibrium existence requires randomization at the resource allocation level, we show that when soldier allocations are continuous, deterministic allocations suffice to guarantee the existence of equilibrium in some cases. 
This yields several practical benefits: such equilibria admit exact representation, are computationally tractable, and, in applications such as security, can be easier to deploy when randomization at the allocation level is impractical or undesirable. 
We view this setting as a natural starting point for analyzing randomized allocation strategies, whose equilibria may lack finite representations and are more difficult to characterize. We leave the analysis of such strategies to future work.

In this setting, players allocate any real-valued fraction of soldiers to each battlefield. Let $\contsoldiers^j=(\contsoldiers^j_1,\cdots,\contsoldiers^j_n)\in\Sigma^j$ be player $j$’s continuous allocation vector. Then the payoff in battlefield $i$ is $u_i(\alpha_i^1,\alpha_i^2,\contsoldiers^1_i,\contsoldiers^2_i)$. 
Considering both the sum and min aggregators, we show that, in the two-sided setting, a partially pure Nash equilibrium (PPNE) -- i.e. one in which players use pure soldier allocation strategies and mixed subgame strategies  -- may fail to exist under either aggregator, even when the utility functions in all battlefields are continuous in the allocation variables.

\begin{restatable}{theorem}{existtwosidedcont}\label{thm:exist-two-sided-cont} 
Consider a two-sided continuous two-level Blotto game under the sum or min aggregator. Even if battlefield utilities are continuous in $\sigma$, a PPNE may not exist.
\end{restatable}

We proceed to analyze the one-sided setting. We show that under the sum aggregator, a PPNE may still fail to exist, whereas under the min aggregator, a PPNE is guaranteed to exist.  We explain this divergence at the end of the section.

\subsection{Sum Aggregator}\label{sect:cont-sum}
In the continuous two-sided model, each player $j$ deterministically selects a soldier allocation $\contsoldiers^j\in\Sigma^j$, and employs a mixed strategy  $\delta^j_{i,\contsoldiers_i^j}$ over actions $\alpha_i^j\in \mathcal A_i^j$ in subgame $G_i$.
We can write the overall  in the two-sided model as:
 \(
U(\delta,\contsoldiers) =\sum_{i} \biggl(\sum_{\alpha_i^1} \sum_{\alpha_i^2} \delta_{i,\contsoldiers_i^1}^1(\alpha_i^1)\delta^2_{i,\contsoldiers_i^2}(\alpha_i^2) u_i(\alpha_i^1,\alpha_i^2,\contsoldiers_i^1,\contsoldiers_i^2) \biggr) .\)
The total payoff is non-convex non-concave in the minimizing and maximizing player's strategies, respectively. This remains true even in the one-sided allocation model. Consequently, a PPNE may still fail to exist, even when each $u_i$ is continuous in the soldier allocation variables.

\begin{restatable}{proposition}{existonesidedsumcont}\label{prop:exist-one-sided-sum-cont} Even in the one-sided case, a continuous two-level Blotto with continuous battlefield utilities may not admit a PPNE under the sum aggregator. The max-min and min-max values are attained but do not coincide.
\end{restatable}

Although equilibrium existence cannot be guaranteed in the one-sided model under general continuous utilities, we show that if each battlefield payoff is linear in the maximizing player’s allocation, namely $u_i\bigl(\contsoldiers_i^2\bigr)=c_i\,\contsoldiers_i^2,\; c_i>0$, a PPNE is guaranteed to exist.
In fact, with linear battlefield utilities, the game can be cast as a bilinear saddle point problem by absorbing the soldier allocation variable $\sigma^2$ and the subgame  distribution over actions $\delta^2$ into a single sequence‑form variable  
$y^2_{\contsoldiers^2}$ via the treeplex construction. The product variable $y^2_{\contsoldiers^2}$ is defined over a polytope $\mathcal Q$ that is represented by the following constraints: \[
\mathcal{Q} = \left\{ 
\begin{array}{ll}
y^2_{\contsoldiers_i^2}(i) \geq 0 & \\
y^2_{\contsoldiers_i^2}(i, \alpha_i^2) \geq 0 & \\
\end{array} \middle| 
\begin{array}{ll}
    y^2(\emptyset) = m^2, \\ 
    \sum\limits_{i \in [n]} y^2_{\contsoldiers_i^2}(i) = y^2(\emptyset), \\ 
    \sum\limits_{\alpha_i^2 \in \mathcal{A}_i^2} y^2_{\contsoldiers_i^2}(i, \alpha_i^2) = y^2_{\contsoldiers_i^2}(i). \\ 
\end{array}
\right\}
\] 
The resulting payoff is bilinear in the variables $y^2_{\contsoldiers^2}$ and $\delta^1$, which allows us to invoke the minimax theorem and guarantee the existence of a PPNE.  

\begin{restatable}
{theorem}{siononesidedsumcontlinear}\label{thm:sion-one-sided-cont-sum}
Consider a one-sided continuous two-level Blotto game with sum aggregator, and suppose $\contsoldiers^2 \in \Sigma^2$, $\delta_i^j\in \Delta(\mathcal A_i^j),j\in\{1,2\}$. Assume that in each battlefield, the utility is linear in the maximizing player's allocation, i.e. $u_i=c_i\cdot\contsoldiers_i^2$ for some constant $c_i>0$. 
Then the minimax theorem applies. 
\end{restatable}

It follows from~\Cref{thm:sion-one-sided-cont-sum} that the continuous two-level Blotto game with sum aggregator and linear utilities admits an equilibrium, which can be computed by solving a linear program (see~\Cref{dual:one-sided-cont-sum}) defined over the polytope $\mathcal Q$.
Since $\mathcal Q$ has a polynomial size representation, we have:

\begin{corollary}\label{cor:exist-one-sided-sum-linear-disc}
    When battlefield utilities are linear in the maximizing player's allocation, a PPNE of the one-sided continuous two-level Blotto game with sum aggregator exists and can be computed in polynomial time.
\end{corollary}

\subsection{Min aggregator}
In the two-sided model, when the overall payoff is defined as the minimum of the expected utilities across battlefields, it can be written as 
\(
U(\delta,\contsoldiers) = \min_{i\in[n]} \biggl(\sum_{\alpha_i^1} \sum_{\alpha_i^2} \delta_{i,\contsoldiers_i^1}^1(\alpha_i^1)\delta^2_{i,\contsoldiers_i^2}(\alpha_i^2) u_i(\alpha_i^1,\alpha_i^2,\contsoldiers_i^1,\contsoldiers_i^2) \biggr) .\)
 
Since a PPNE may fail to exist in the two-sided setting (\Cref{thm:exist-two-sided-cont}), we focus on the one-sided model. We show that equilibrium non-existence arises only if the utility function in every battlefield $i$ is discontinuous in $\sigma^2_i$.

\begin{restatable}{theorem}{existonesidedmindisc}\label{thm:exist-one-sided-min-cont} Consider a one-sided continuous two-level Blotto game with min aggregator. Then a PPNE may not exist when battlefield utilities are discontinuous. 
\end{restatable}

However, when the payoff in each battlefield is a continuous function of the soldier allocation, the continuous two-level Blotto game with min aggregator admits a saddle point in the one-sided formulation: the max-min and min-max values coincide and are attained, so a PPNE exists. We establish this by invoking the Kneser-Fan minimax theorem.
Suppose that only Player 2 allocates soldiers to battlefields, and both players employ mixed strategies to engage in the battlefield subgames. Then, the min-max formulation of the continuous two-level Blotto game is:
\[
\begin{aligned}
  &\min_{\delta^1}\max_{\delta^2,\contsoldiers^2}\min_{i\in[n]}
    \mathbb E\bigl[u_i(\alpha_i^1,\alpha_i^2,\contsoldiers_i^2)\bigr] \\
    &=\min_{\delta^1}\max_{\delta^2,\contsoldiers^2}\min_{i\in[n]} 
    \sum_{\alpha_i^1}\sum_{\alpha_i^2}
    \delta_i^1(\alpha_i^1)\,\delta^2_{i,\contsoldiers_i^2}(\alpha_i^2)
    u_i(\alpha_i^1,\alpha_i^2,\contsoldiers_i^2).
\end{aligned}
\]
Ideally, we would like to invoke Sion’s minimax theorem to interchange the outer minimization and maximization, thereby recovering the equivalent max-min formulation of the game and establishing the minimax identity. However, the inner objective function cannot be made simultaneously quasiconvex in the minimizing player’s strategy and quasiconcave in the maximizing player’s strategy under any reordering of the operators and/or regrouping of the variables. Consequently, Sion’s conditions are not met. Therefore, we appeal to the Kneser-Fan minimax theorem (\Cref{thm:kneser-fan}). 
Taking the dual of the inner minimization problem over $i\in[n]$, allows us to obtain the following result:

\begin{restatable}{theorem}{siononesidedmincont}\label{thm:sion-one-sided-min-cont} Consider a one-sided continuous two-level Blotto game with continuous battlefield utilities. Then the minimax theorem holds. Specifically,
    \(
    \min_{\delta^1} \max_{\delta^2,\contsoldiers^2} \min_{i} \mathbb E[u_i(\alpha_i^1,\alpha_i^2,\contsoldiers_i^2)] =
    \max_{\delta^2,\contsoldiers^2} \min_{\delta^1} \min_{i} \mathbb E[u_i(\alpha_i^1,\alpha_i^2,\contsoldiers_i^2)].
    \)
\end{restatable}   

\begin{remark}
    This approach does not extend to any of the other settings we analyzed in the continuous domain, which is consistent with the non‑existence results established in~\Cref{thm:exist-two-sided-cont} and~\Cref{prop:exist-one-sided-sum-cont}. In particular, (1) in the one‑sided continuous setting with sum aggregator, the payoff fails to be concavelike in the maximizing player’s strategy, which immediately extends to the sum-aggregated two-sided setting; whereas (2) in the two-sided continuous setting with min aggregator, convexlikeness of the payoff in the minimizing player’s strategy cannot be guaranteed. A formal discussion is in~\Cref{discussion:KF}. 
\end{remark}

For the special case where battlefield utilities are linear in $\contsoldiers^2$, we obtain a sequence-form representation of the strategies of both players defined over the polytopes $\mathcal P$ and $\mathcal Q$. Hence, we have the following theorem:
\begin{restatable}
{theorem}{componesidedmincontlinear}\label{thm:ne-comp-one-sided-cont-min}
    Consider a one-sided continuous two-level Blotto game under the min aggregator, and suppose $\contsoldiers^2 \in \Sigma^2$, $\delta_i^j\in \Delta(\mathcal A_i^j),j\in\{1,2\}$. Assume that in every battlefield, the utility is linear in the maximizing player's allocation, i.e. $u_i=c_i\cdot\contsoldiers_i^2$ for some constant $c_i>0$. 
    Then we can compute a PPNE of the game by solving a linear program.
\end{restatable}

\begin{restatable}{corollary}{existonesidedminlinear}\label{cor:one-sided-min-linear-cont}
When battlefield utilities are linear in the soldier allocation of the maximizing player, the one-sided continuous two-level Blotto game with min aggregator admits a PPNE that is polynomial-time computable.
\end{restatable}

\paragraph{Computing the max-min strategy}
In this paragraph, we focus on computing the max-min strategy for Player 2 in the one-sided continuous two-level Blotto game with min aggregator. We show that it can be reformulated as a maximization over the minimum Nash value across battlefields and hence it can be computed using a subgradient ascent algorithm. 

\begin{restatable}{proposition}{maxminonesidedmincont}\label{thm:maxmin-one-sided-min-cont-linear} Consider a one-sided continuous two-level Blotto game with min aggregator, and suppose $\contsoldiers^2 \in \Sigma^2$, $\delta_i^j\in \Delta(\mathcal A_i^j),j\in\{1,2\}$. Then the maxmin problem can be formulated as a maximization problem over the minimum Nash value across battlefields. Specifically, 
    \(
   \max_{\delta^{2},\sigma^2}\min_{\delta^{1}}\min_i \mathbb{E}[u_i(\alpha_i^{1},\alpha_i^{2},\sigma^2)] = \max_{\sigma^2}\min_{i\in[n]}v^*_i(\contsoldiers_i^2)
   \).  
\end{restatable}

Let $V(\sigma^2):=\min_i v^*_i(\sigma^2)$ denote the objective function of this maximization problem. We show that this function is quasiconcave. Notice that, contrary to the sum, the min aggregator preserves quasiconvcavity. 
Hence, to achieve overall quasiconcavity of $V(\contsoldiers^2)$, it suffices to have that the individual terms $v^*_i(\contsoldiers)$ are quasiconcave.

\begin{restatable}{lemma}{quasiconcavitymaxmin}
    \label{lemma:quasiconcavity} If for each battlefield $i\in [n]$, $u_i(\alpha_i^1,\alpha_i^2,\sigma^2)$ is increasing in $\sigma^2_i$ for all $\alpha_i^1\in\mathcal A_i^1$ and $\alpha_i^2\in\mathcal A_i^2$, then the aggregate Nash value function $V(\sigma^2)=\min_i v^*_i(\sigma^2)$ is quasiconcave.
\end{restatable}

Since $V(\sigma^2)$ is quasiconcave on the compact convex set of soldier assignments $\Sigma^2$, then finding the optimal $\sigma^{2*}= \arg\max_{\sigma^2\in \Sigma^2}V(\sigma^2)$ reduces to maximizing a quasiconcave function over a convex set.  Moreover, under the assumption that each battlefield utility function $u_i(\sigma^2)$ is Lipschitz continuous on the soldiers assignments simplex $\Sigma^2$, not only does a global maximizer $\sigma^{2*}$ exist, but one can actually compute it using a simple projected subgradient‐ascent method. 
Specifically, at each iteration $t$, we select a subgradient $g_t\in \partial V(\sigma_t^2)$, a step size $\eta_t>0$ and perform the update \(
  \sigma^2_{t+1}
  = P_{\Sigma^2}\bigl(\sigma^2_t + \eta_t\,g_t\bigr),
\)
where \(P_{\Sigma^2}\) denotes the projection onto Player 2's soldiers' simplex. For completeness, the full pseudocode is given in~\Cref{app:algo}.
Next, we derive a closed‐form description of the subdifferential \(\partial V(\sigma_t^2)\), which enables efficient computation of each $g_t$.

\begin{restatable}
{proposition}{subgnashval}\label{prop:nash-subgrad}
    Consider the two‐player zero‐sum subgame $G_i$ on battlefield $i$. Suppose that  for every pure profile $(\alpha_i^1,\alpha_i^2)\in\mathcal A_i^1\times\mathcal A_i^2$, the payoff $u_i(\alpha_i^1,\alpha_i^2,\sigma^2_i)$ is Lipschitz continuous in $\sigma^2_i$. Let $(\delta_i^{1*},\delta_i^{2*})$ be any NE of $G_i$ at a given $\sigma^2_i$.  Then the vector $g_i=
    \sum_{\alpha_i^1\in\mathcal A_i^1}\sum_{\alpha_i^2\in\mathcal A_i^2}
      \delta_i^{1*}(\alpha_i^1)\;\delta_i^{2*}(\alpha_i^2)\;
      z_{\alpha_i^1,\alpha_i^2}, \quad z_{\alpha_i^1,\alpha_i^2}\in\partial_{\sigma^2_i}u_i(\alpha_i^1,\alpha_i^2,\sigma^2_i),$ is a subgradient of $v_i^*(\sigma^2)$.
\end{restatable}

\begin{restatable}
{corollary}{subgaggfct}\label{cor:V-subgrad}
  Under the assumptions of~\Cref{prop:nash-subgrad}, $V$ is Lipschitz on $\Sigma^2$ and its subdifferentials $\partial V(\sigma^2)$ lie in 
    \(
    \conv\Bigl\{
      \partial v_i^*(\sigma^2):
      i\in\arg\min_{j\in[n]}v_j^*(\sigma^2)
    \Bigr\}.
  \)
\end{restatable}

Finally, to obtain the optimal soldier allocation \(\sigma^{2*}\) for Player 2, we run Algorithm~\ref{alg:psa}.  Each iteration \(t\) requires a subgradient \(g_t\in\partial V(\sigma^2_t)\).  By Corollary~\ref{cor:V-subgrad}, it suffices to pick any active battlefield $i^* \in \arg\min_{i} v_i^*(\sigma^2_t)$ and use its subgradient. \Cref{prop:nash-subgrad} then provides a concrete ``Nash subgradient'' for battlefield $i^*$: choose for each action pair \((\alpha_i^1,\alpha_i^2)\) a subgradient 
$z_{\alpha_i^1,\alpha_i^2}\in\partial_{\sigma^2_{i^*}}u_{i^*}(\alpha_i^1,\alpha_i^2,\sigma^2_{i^*}),$
and set $g_t
  = \sum_{\alpha_i^1,\alpha_i^2}
    \delta_{i^*}^{1*}(\alpha_i^1)\,\delta_{i^*}^{2*}(\alpha_i^2)\,
    z_{\alpha_i^1,\alpha_i^2}
  \;\in\;\partial v_{i^*}^*(\sigma^2_t).$
This \(g_t\) is exactly the update direction used in the projected subgradient ascent algorithm (Algorithm~\ref{alg:psa} in~\Cref{app:algo}).

\paragraph{Convergence Guarantees} 
By~\Cref{lemma:quasiconcavity} and~\Cref{cor:V-subgrad}, our aggregate Nash‐value function 
$V$ is quasiconcave and Lipschitz continuous on the soldiers simplex.  Hence, the projected subgradient‐ascent iterates $\{\contsoldiers^2_t\}$ satisfy the general basic‐inequality conditions (H1)-(H2) of~\cite{hu2020convergence}, which guarantee global convergence of any sequence with a suitably chosen step size rule. Moreover, they establish in Theorems 3.3-3.5 that, under an additional \emph{weak sharp minima of Hölderian order} assumption and upper‐bounded noise, such a sequence converges at a linear rate when the Hölder exponent equals 1.  In our maximization context, this translates into a \emph{weak sharp maxima} condition:

\begin{figure*}[t]
    \centering
    \begin{subfigure}[h]{0.32\linewidth}
    \centering
    \includegraphics[width=\textwidth]
    {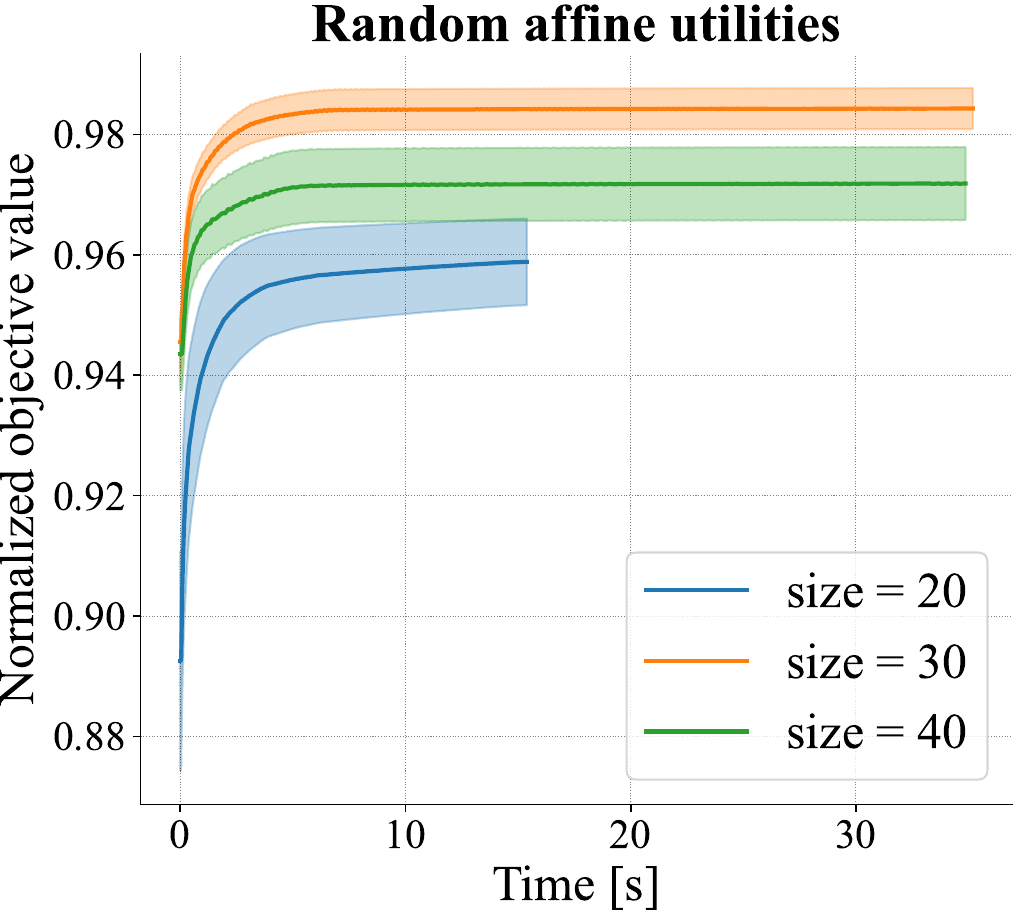}
\end{subfigure}
\hfill
\begin{subfigure}[h]{0.32\linewidth}
    \centering
    \includegraphics[width=\textwidth]{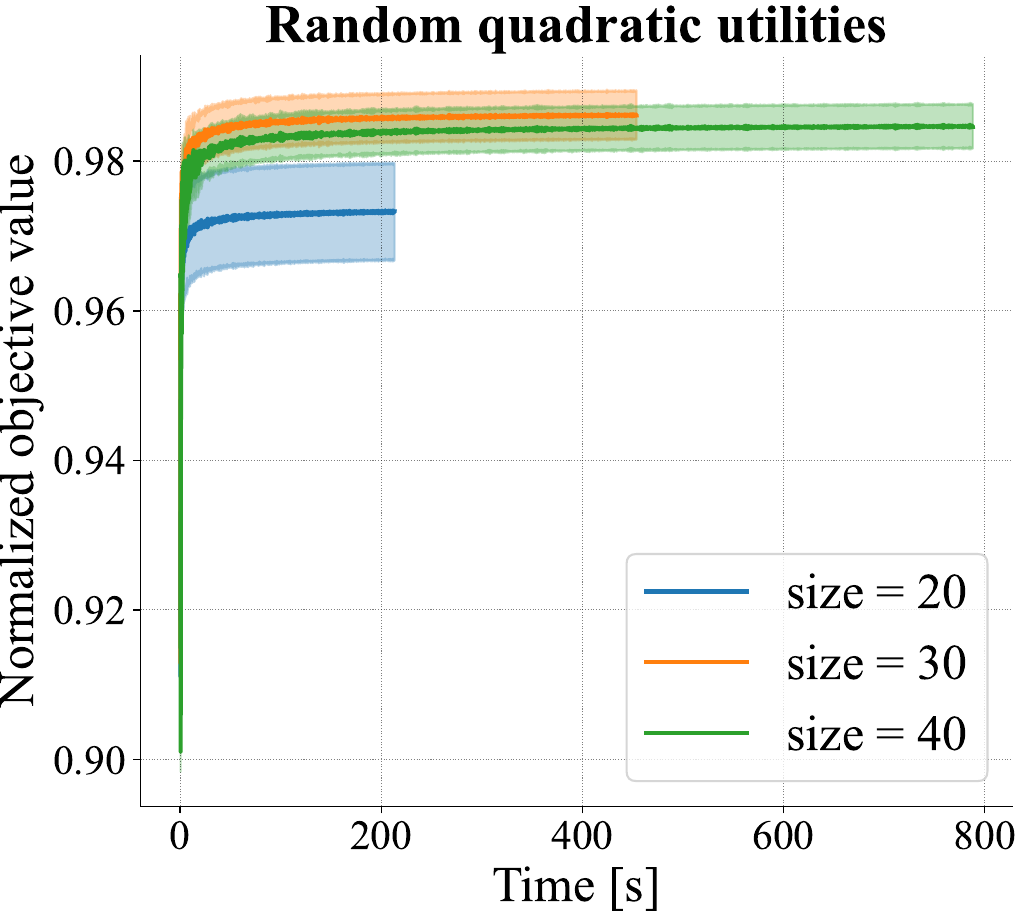}
\end{subfigure}
    \hfill
\begin{subfigure}[h]{0.35\linewidth}
    \centering
    \includegraphics[width=\textwidth]{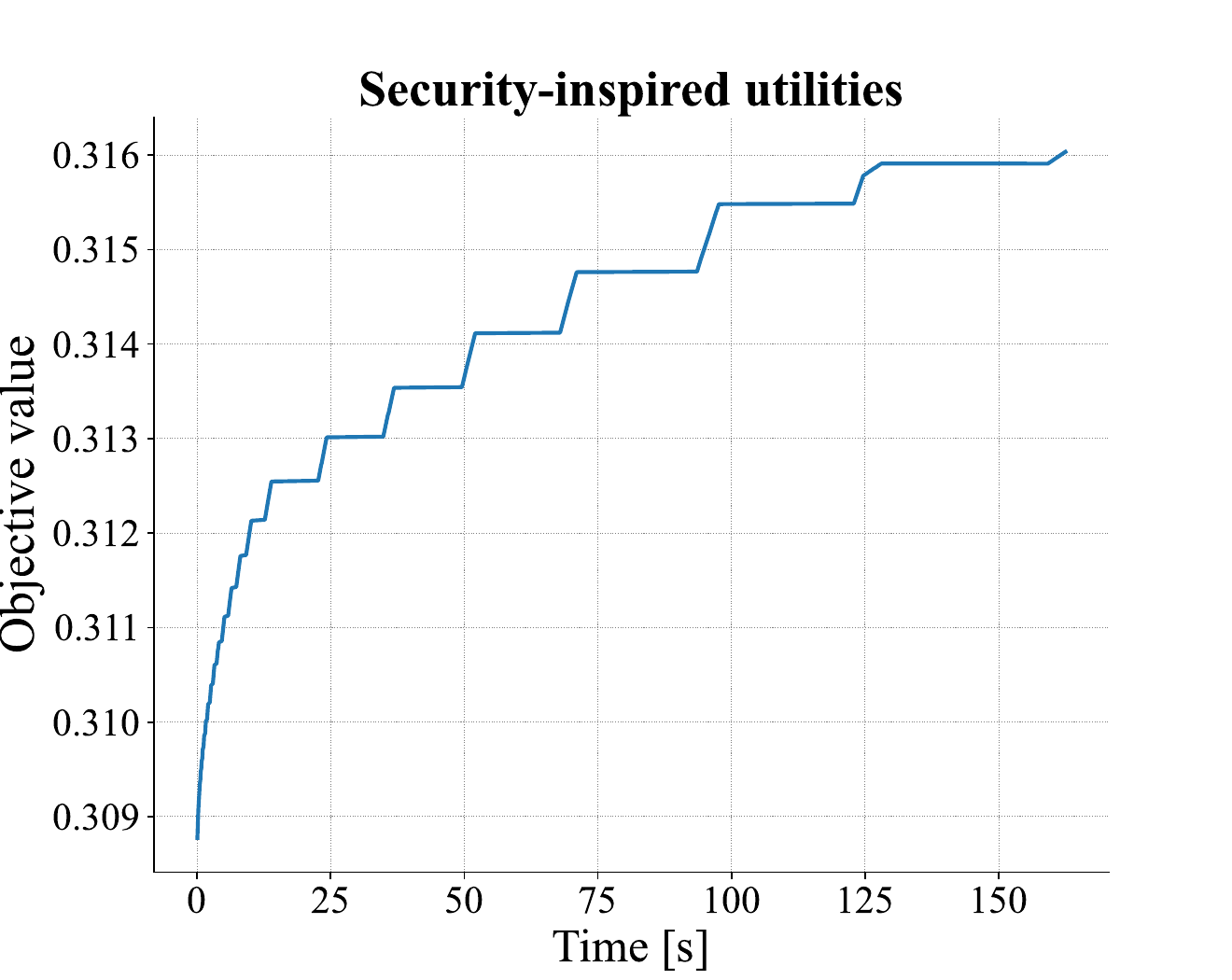}
\end{subfigure}
    \caption{Convergence of~\Cref{alg:psa} in the one-sided min-aggregated setting with affine (left), quadratic (middle), and security-inspired (right) battlefield utilities.}
    \label{fig:conv_subgr_asc}
\end{figure*}

\begin{definition}[Weak Sharp Maxima]
There exist constants $\epsilon>0$, $p\in(0,1]$ such that $\forall \contsoldiers^2\in\Sigma^2$,
\(
V(\contsoldiers^{2*}) - V(\contsoldiers^2) \;\ge\;\epsilon\,\|\contsoldiers^2-\contsoldiers^{2*}\|^p,
\)
where $\contsoldiers^{2*}$ is the unique maximizer of $V$.
\end{definition}
When battlefield utilities are affine functions of $\contsoldiers^2$, the induced Nash value $v_i^*(\contsoldiers_i^2)$ remains an affine function of $\contsoldiers_i^2$. It follows that the aggregate function $V=\min_i v_i^*$ is concave on the soldiers simplex $\Sigma^2$. Applying the subgradient inequality at the maximizer $\contsoldiers^{2*}\in \arg\max_{\Sigma^2}V$ then yields the following sharp‐max error bound for $V$.  

\begin{restatable}{proposition}{sharpmax}\label{prop:sharp-max} 
    If every battlefield utility is affine in Player 2’s soldier allocation, i.e. $u_i(\sigma^2)=c_i\,\sigma_i^2+d_i$ with $c_i>0$, then $V$ satisfies the weak sharp‐maxima condition with exponent $p=1$. 
\end{restatable}

It follows immediately from Theorem 3.3 of~\cite{hu2020convergence} that the subgradient ascent  iterates converge linearly to the optimal  $\contsoldiers^{2*}$.
\section{Empirical Evaluation}

\paragraph{Discrete two-sided with sum-aggregator} We first evaluate in the most basic setting of discrete soldiers under the sum aggregator (Section~\ref{sec:sum-agg}). We consider (i) the LP-based approach similar to \citet{behnezhad2023fast} and (ii) our approach based on online learning. The $i\in[n]$ battlefield has value $i$, such that winning (resp. losing) it gives $+i$ (resp. $-i$) reward. For consistency, we re-normalize values of battlefields such that they sum to $1$. The $i$-th battlefield is won by player $1$ with probability $k^1_i/(k^1_i + k^2_i)$, with a random winner selected if no soldiers are allocated. For each subgame, players decide whether to \textit{double} their stakes. If just one player doubles, payoffs and losses for that battlefield are doubled; if both players double, they are quadrupled.

We report running times for both methods. For the former, we use the default configuration for Gurobi~\cite{gurobi}. For the latter, we adopt the Regret Matching-Plus (RM+) algorithm~\cite{tammelin2015cfrplus} 
and construct $\Gamma^j$ using scaled extensions \citep{farina2019efficient}. We declare convergence when the saddle-point gap is less than $0.002$. More details on experimental setup are reported in Appendix~\ref{sec:appendix-expts}. The results are shown in~\Cref{tab:expt-rm}. We observe that when the game is small, the LP finds the NE easily. However, as the size of the game grows, LPs slow down dramatically and online learning approaches begin to shine. In fact, for larger instances Gurobi's sometimes fails to converge because of numerical issues or simply going out-of-memory.\footnote{In our experiments, it is unclear if Barrier or Primal/Dual simplexes are superior. For larger games, it would appear as though simplex performs better in practice, while barrier tends to lead to non-convergence.} In contrast, methods based on online learning scale better.
\begin{table}[h]
    \centering
    \setlength\tabcolsep{1mm}
    \begin{tabular}{@{}cccc|ccc@{}}
        \toprule
        $n$ & $m^1$ & $m^2$ & $|\Gamma^1|, |\Gamma^2|$ & LP [s] & RM+ [s] \\
        \midrule
        30 & 100 & 50 & 85, 49 & 91 & 0.02
        \\
        35 & 125 & 70 & 281k, 92k & $4.8\cdot 10^3$ & 121 
        \\
        40 & 150 & 100 & 1M, 262k& $5.4 \cdot 10^3$ & 525 
        \\
        50 & 200 & 100 & 12.6M, 2.05M & NA ($>2.6 \cdot 10^4$) & $9.5\cdot 10^3$ 
        \\
        \bottomrule
    \end{tabular}
    \caption{Runtime for discrete soldiers with sum aggregators.}
    \label{tab:expt-rm}
\end{table}


\paragraph{Continuous one-sided with min-aggregator} We evaluate Algorithm~\ref{alg:psa} on (1) randomly generated and (2) security-inspired battlefield utilities. For random instances, we consider (i) affine utilities $u_i(\contsoldiers^2)=c_i\,\contsoldiers^2+d_i$, and (ii) quadratic utilities $u_i(\contsoldiers^2)=b_i\,(\contsoldiers^2)^2+c_i\,\contsoldiers^2+d_i$, with $b_i,c_i,d_i\sim\mathrm{Uniform}[0,100]$. Player 2 is given 20 soldiers to allocate across 5 battlefields. We vary players' action spaces in $\{20,30,40\}$ and generate 10 independent instances per game size. In the first 2 plots of~\Cref{fig:conv_subgr_asc}, we show the aggregate normalized value $\bar V^t$ and its standard error 
$\mathrm{SE}(t)
= \sqrt{\frac{s_t}{\ell}},
\quad
s_t = \frac{1}{\ell-1}\sum_{k=1}^\ell\bigl(V_k^t - \bar V^t\bigr)^2,$
across $\ell=10$ instances, where $V_k^t$ denotes the normalized objective value on instance $k$ at iteration $t$. Under both utility functions, all three sizes reach near‐optimality with very low dispersion. 
Next, we consider a real-world scenario inspired by security applications. 
We model a one‐sided min-aggregated continuous two‐level Blotto game where Player 2 allocates 10 soldiers across three battlefield, each featuring a two‐player zero-sum security subgame where Player 2 is the defender’s and Player 1 the attacker. We follow the framework of~\cite{krever2025guardconstructingrealistictwoplayer}. More details on the experimental setup are provided in~\Cref{app:cont-one-sided-experiments}. The rightmost plot of~\Cref{fig:conv_subgr_asc} confirms the algorithm’s robustness in this practical setting.

\section{Conclusion}


In this paper, we introduced and analyzed a two-level variant of the Colonel Blotto game, where players allocate resources across battlefields and then engage in parametrized subgames. We examined multiple modeling dimensions -- soldier types, payoff aggregation methods, and two-sided or one-sided strategic settings -- and established equilibrium existence, provided solution algorithms, and validated their practicality through experiments.

\section*{Acknowledgments}
This research was supported by the Office of Naval Research awards N00014-22-1-2530 and N00014-23-1-2374, the National Science Foundation awards IIS-2147361 and IIS-2238960, the Ministry of Education, Singapore, under the Academic Research Fund Tier 1 (FY2025) and by the National University of Singapore, under the Start-Up Grant Scheme. 
We thank Noah Krever for help with setting up the security-inspired subgames.

\bibliography{aaai2026}

\newpage
\onecolumn
\appendix
\section{Notation}\label{app:notation}

\begin{table}[h]
    \centering
    \begin{tabular}{@{}cp{8.5cm}p{4.5cm}@{}}
        \toprule
        \textbf{Symbol} & \textbf{Description} & \textbf{Strategy Representation} \\ 
        \midrule
        
        $\gamma^{j}$ & distribution over soldier assignments by player $j$ & mixed  strategy over soldiers\\ 
        \midrule
        $\delta_i^{j}$ & distribution over actions played by player $j$ in subgame $G_i$ & mixed strategy over subgame actions \\ 
        \midrule
        
        $h^{j}_{i,a,b}$ & flow from battlefield $i$ for player $j$, starting from $a$ soldiers, and placing $a-b$ soldiers into battlefield $i$ & flow variable \\
        \midrule
        $x^j_{i,\alpha_i^j,k_i^j}$ & probability that player $j$ places $k_i^j$ discrete soldiers into battlefield $i$ and plays action $\alpha_i^j$ there & mixed strategy (over flow polytope)\\
        \midrule 
        $y^j({i,\alpha_i^j})$ & product of the probability that player $j$ plays action $\alpha_i^j$ in battlefield $i$ with another variable &  sequence form product representation\\
        \bottomrule
    \end{tabular}
    \caption{Notation}
    \label{tab:notation}
\end{table}
\section{Linear programs for computing a Nash equilibrium in some settings}

\subsection{Two-sided discrete two-level Blotto game with sum aggregator}\label{dual:two-sided-disc-sum}
The best response of Player 1 to a fixed strategy $x^{2}_{i,k,\alpha_i^{2}}$ of Player 2 is 
\begin{align*}
\min \sum_{i=1}^{n} \sum_{\alpha_i^{1} \in \mathcal{A}^{1}_i} \sum_{\alpha_i^{2} \in \mathcal{A}^{2}_i}
    \sum_{k^{1}\in[m^{1}]}
    \sum_{k^{2}\in[m^{2}]}
    &u_i (\alpha_i^{1}, \alpha_i^{2}, k_i^{1}, k_i^{2}) \cdot x^{1}_{i,k_i^{1},\alpha_i^{1}}
    \cdot x^{2}_{i,k_i^{2},\alpha_i^{2}} \\
    h^{1}_{0, m^{1}, m^{1}} &= 1 \\ 
    h^{1}_{0, a, b} &= 0 &&\forall a,b\in[0,m^{1}],  (a, b) \neq (m^{1},m^{1}) \\
    h^{1}_{n-1, 0, 0} & = 1 \\
    \sum_{a=c}^{m^{1}} h^{1}_{i-1,a,c} &= \sum_{b=0}^{c} h^{1}_{i, c, b} &&\forall i\in [1,n-1],c\in[0, m^{1}] \\ 
    \sum_{\alpha_i^1 \in \mathcal{A}^{1}_i} x^{1}_{i,k_i^{1},\alpha_i^1} &= \sum_{k_i^1\leq r\leq m^{1}} h^{1}_{i, r, r -k_i^{1}} &&\forall i\in [1,n-1],k_i^1\in [0,m^{1}] \\
    x^{1}_{i,k_i^{1},\alpha_i^1} & \geq 0 \\
    h^{1}_{i,a,b} &\geq 0
\end{align*}

The dual of the above LP can be written as follows

\begin{align*}
    \max&~\lambda_{0, m^1, m^1} + \lambda_{n-1,0,0}\\
    0 &\geq 1[j\geq k]\left(\lambda_{i+1,k} - \lambda_{i,j} - \mu_{i,j-k}\right) + F(i,j,k)\tag{$h^{1}_{i,j,k}$}\\
    \mu_{i,k}&\leq \sum_{\alpha_i^{2} \in \mathcal{A}^{2}_i}
    \sum_{k_i^{2}\in[m^{2}]}
    u_i^{1} (\alpha_i^{1}, \alpha_i^{2}, k_i^{1}, k_i^{2})
    \cdot x^{2}_{i,k_i^2,\alpha_i^{2}} \tag{$x^{1}_{i,k_i^1,\alpha_i^{1}}$}
\end{align*}

where
\begin{equation*}
    F(i,j,k)=
    \begin{cases}
    \lambda_{0, m^{1}, m^{1}} & \text{if}~i=0,j=k=m^{1}\\
    \lambda_{0, j, k} & \text{if}~i=0,(j,k)\neq (m^{1},m^{1})\\
    \lambda_{n-1,0,0} & \text{if}~i=n-1, j=k=0\\
    0 & \text{otherwise}.
    \end{cases}
\end{equation*}

Since the game is zero-sum, we can find a Nash equilibrium of the two-level Blotto game by appending the best response constraints for Player 2. Therefore, we obtain the following LP

\begin{align*}
    \max&~\lambda_{0, m^{1}, m^{1}} + \lambda_{n-1,0,0}\\
    0 &\geq 1[j\geq k]\left(\lambda_{i+1,k} - \lambda_{i,j} - \mu_{i,j-k}\right) + F(i,j,k)\\\
    \mu_{i,k}&\leq \smashoperator{\sum_{\alpha_i^{2} \in \mathcal{A}^{2}_i, k_i^{2}\in[m^{2}]}}
    u_i^{1} (\alpha_i^{1}, \alpha_i^{2}, k_i^{1}, k_i^{2})
    \cdot x^{2}_{i,k_i^2,\alpha_i^{2}} \\
    h^{2}_{0, m^{2}, m^{2}} &= 1 \\ 
    h^{2}_{0, a, b} &= 0 \quad\forall a,b\in[0,m^{2}],  (a, b) \neq (m^{2},m^{2})  \\
    h^{2}_{n-1, 0, 0} & = 1 \\
    \sum_{a=c}^{m^{2}} h^{2}_{i-1,a,c} &= \sum_{b=0}^{c} h^{2}_{i, c, b} \quad\forall i\in [1,n-1],c\in[0, m^{2}]& \\ 
    \sum_{\alpha_i^2 \in \mathcal{A}^{2}_i} x^{2}_{i,k_i^2,\alpha_i^2} &= \sum_{k_i^2\leq r\leq m^{2}} h^{2}_{i, r, r -k_i^2} \quad\forall i\in [1,n-1],k_i^2\in [0,m^{2}] \\
    x^{2}_{i,k_i^2,\alpha_i^2},h^{2}_{i,a,b} & \geq 0 \\
    \lambda,\mu &\in\mathbb{R}
\end{align*}

\subsection{One-sided discrete two-level Blotto game with min aggregator}\label{dual:one-sided-disc-min}
\paragraph{Max-min strategy}
The best response of Player 1 to a fixed strategy $x^{2}_{i,k,\alpha_i^{2}}$ of Player 2 is 
\begin{align*}
\min_{y^1}\sum_{i}\sum_{\alpha^{1}_i,\alpha_i^{2}}\sum_{k_i^{2}}&u_i(\alpha_i^{1}, \alpha_i^{2},\discsoldiers_i^{2})\cdot x^2_{i,k_i^{2},\alpha_i^{2}}\cdot y^1(i,\alpha_i^1) \\
    \sum_{i\in [n]}y^1(i) &= 1\\
    \sum_{\alpha_i^1\in \mathcal A_i^1}y^1(i,\alpha_i^1) &= y^1(i) \quad \forall i \\
    y^1(i),y^1({i,\alpha_i^1}) &\geq 0 \quad \forall i, \alpha_i^1
\end{align*}
We write the dual of the best response LP as follows
\begin{align*}
\max_{\lambda,\mu_i} &~\lambda  \\
    \lambda&- \mu_i\leq 0\\
    \mu_{i} &\leq  \sum_{\alpha_i^2}\sum_{k^2_i} u_i(\alpha_i^1,\alpha_i^2,k_i^2)\cdot x^2_{i,k_i^2\alpha^2_i} \quad \forall i,\alpha_i^1
\end{align*}

Therefore, a NE of the two-level Blotto game under the one-sided case can be computed by solving the following LP
\begin{align*}
   \max \lambda  \\
    \lambda&- \mu_i\leq 0\\
    \mu_{i} &\leq  \sum_{\alpha_i^2}\sum_{k^2_i} u_i(\alpha_i^1,\alpha_i^2,k_i^2)\cdot x^2_{i,k_i^2\alpha^2_i} \quad \forall i,\alpha_i^1 \\
    h_{0, m^{2}, m^{2}} &= 1\\ 
    h_{0, a, b} &= 0 &&\forall a,b\in[0,m^{2}],  (a, b) \neq (m^{2},m^{2}) \\
    h_{n-1, 0, 0} & = 1\\
    \sum_{a=c}^{m^{2}} h_{i-1,a,c} &= \sum_{b=0}^{c} h_{i, c, b} &&\forall i\in [1,n-1],c\in[0, m^{2}] \\ 
    \sum_{\alpha_i^2 \in \mathcal{A}^{2}_i} x^2_{i,k_i^2,\alpha_i^2} &= \sum_{k_i^2\leq r\leq m^{2}} h_{i, r, r -k_i^2} &&\forall i\in [1,n-1],k_i^2\in [0,m^{2}] \\
    0 &\leq x^2_{i,k_i^2,\alpha_i^2}, h_{i,a,b}
\end{align*}

\paragraph{Min-max strategy}
The best response of Player 2 to a fixed strategy $y^1(i,\alpha_i^1)$ of Player 1 is 
\begin{align*}
\max_{x^2}\sum_{i}\sum_{\alpha^{1}_i,\alpha_i^{2}}\sum_{k_i^{2}}&u_i(\alpha_i^{1}, \alpha_i^{2},\discsoldiers_i^{2})\cdot x^2_{i,k_i^{2},\alpha_i^{2}}\cdot y^1(i,\alpha_i^1) \\
    h^{2}_{0, m^{2}, m^{2}} &= 1 \\ 
    h^{2}_{0, a, b} &= 0 &&\forall a,b\in[0,m^{2}],  (a, b) \neq (m^{2},m^{2}) \\
    h^{2}_{n-1, 0, 0} &= 1 \\
    \sum_{a=c}^{m^{2}} h^{2}_{i-1,a,c} &= \sum_{b=0}^{c} h^{2}_{i, c, b} &&\forall i\in [1,n-1],c\in[0, m^{2}] \\ 
    \sum_{\alpha_i^2 \in \mathcal{A}^{2}_i} x^{2}_{i,k_i^{2},\alpha_i^2} &= \sum_{k_i^1\leq r\leq m^{1}} h^{2}_{i, r, r -k_i^{2}} &&\forall i\in [1,n-1],k_i^1\in [0,m^{2}] \\
    x^{2}_{i,k_i^{2},\alpha_i^1} & \geq 0 \\
    h^{2}_{i,a,b} &\geq 0
\end{align*}
The dual of the above LP can be written as follows

\begin{align*}
    \min&~\lambda_{0, m^2, m^2} + \lambda_{n-1,0,0}\\
    0 &\leq 1[j\geq k]\left(\lambda_{i+1,k} - \lambda_{i,j} - \mu_{i,j-k}\right) + F(i,j,k)\\
    \mu_{i,k}&\geq \sum_{\alpha_i^{1} \in \mathcal{A}^{1}_i}
    u_i (\alpha_i^{1}, \alpha_i^{2}, k_i^{1}, k_i^{2})
    \cdot y^1(i,\alpha_i^1)
\end{align*}

where
\begin{equation*}
    F(i,j,k)=
    \begin{cases}
    \lambda_{0, m^{2}, m^{2}} & \text{if}~i=0,j=k=m^{2}\\
    \lambda_{0, j, k} & \text{if}~i=0,(j,k)\neq (m^{2},m^{2})\\
    \lambda_{n-1,0,0} & \text{if}~i=n-1, j=k=0\\
    0 & \text{otherwise}.
    \end{cases}
\end{equation*}

Therefore, a NE of the two-level Blotto game under the one-sided case can be computed by solving the following LP
\begin{align*}
    \min&~\lambda_{0, m^2, m^2} + \lambda_{n-1,0,0}\\
    0 &\leq 1[j\geq k]\left(\lambda_{i+1,k} - \lambda_{i,j} - \mu_{i,j-k}\right) + F(i,j,k)\\
    \mu_{i,k}&\geq \sum_{\alpha_i^{1} \in \mathcal{A}^{1}_i}
    u_i (\alpha_i^{1}, \alpha_i^{2}, k_i^{1}, k_i^{2})
    \cdot y^1(i,\alpha_i^1) \\
    1 &= \sum_{i\in [n]}y^1(i)\\
    y^1(i) &= \sum_{\alpha_i^1\in \mathcal A_i^1}y^1(i,\alpha_i^1) \quad \forall i \\
    0 &\le y^1(i), y^1({i,\alpha_i^1}) \quad \forall i, \alpha_i^1
\end{align*}

\subsection{One-sided continuous two-level Blotto game with sum aggregator under linear battlefield utilities} \label{dual:one-sided-cont-sum}
\paragraph{Max-min strategy}
The best response of Player 1 to a fixed strategy $y^{2}$ of Player 2 is 
\begin{align*}
\min_{\delta^{1}} \sum_{i\in[n]} \sum_{\alpha_i^1\in \mathcal A_i^1} \sum_{\alpha_i^2\in \mathcal A_i^2}& \delta^1_i(\alpha_i^1)y^2_{\contsoldiers_i^2}(i,\alpha_i^2) c_i \\
    \sum_{\alpha_i^1\in \mathcal A_i^1}\delta_i^1(\alpha_i^1) &= 1 \quad \forall i \\
    \delta_i^1(\alpha_i^1)&\geq 0 \quad \forall i, \alpha_i^1
\end{align*}
We write the dual of the best response LP as follows
\begin{align*}
\max_{\mu_i}& \sum_{i\in[n]}\mu_i  \\
    \mu_{i} &=y^2_{\contsoldiers_i^2}(i,\alpha_i^2) c_i \quad \forall i
\end{align*}

Therefore, a NE of the two-level Blotto game under the one-sided case can be computed by solving the following LP
\begin{align*}
   \max &\sum_{i\in[n]}\mu_i  \\
    y^2_{\contsoldiers_i^2}(i,\alpha_i^2) c_i &= \mu_{i} \quad \forall i \\
    \sum_{i \in [n]} y^2_{\contsoldiers_i^2}(i) &= m^2 \\
    \sum\limits_{\alpha_i^2 \in \mathcal{A}_i^2} y^2_{\contsoldiers_i^2}(i, \alpha_i^2) &= y_{\contsoldiers_i^2}^2(i) \qquad \forall i\\
    y^2_{\contsoldiers_i^2}(i),y^2_{\contsoldiers_i^2}(i, \alpha_i^2) &\geq 0 \qquad \forall i,\alpha_i^2
\end{align*}

\paragraph{Min-max strategy}
The best response of Player 2 to a fixed strategy $\delta^{1}$ of Player 1 can be written as
\begin{align*}
    \max_{y^2} \sum_{i\in[n]} \sum_{\alpha_i^1\in \mathcal A_i^1} \sum_{\alpha_i^2\in \mathcal A_i^2} &\delta^1_i(\alpha_i^1)y^2_{\contsoldiers_i^2}(i,\alpha_i^2) c_i \\
    \sum_{i \in [n]} y^2_{\contsoldiers_i^2}(i) &= m^2 \\
    \sum\limits_{\alpha_i^2 \in \mathcal{A}_i^2} y^2_{\contsoldiers_i^2}(i, \alpha_i^2) &= y_{\contsoldiers_i^2}^2(i)  \qquad \forall i\\
     y^2_{\contsoldiers_i^2}(i),y^2_{\contsoldiers_i^2}(i, \alpha_i^2) &\geq 0 \qquad \forall i,\alpha_i^2
\end{align*}
We write the dual of the best response LP as follows
\begin{align*}
\min_{\lambda,\mu_i} &~\lambda m^2  \\
    0&\leq \lambda - \mu_i\\
    \mu_{i} &\geq  \sum_{\alpha_i^1} \delta^1_i(\alpha_i^1)c_i
\end{align*}
Finally, we obtain a NE of the full Blotto game by solving the following linear program
\begin{align*}
\min &~\lambda m^2  \\
    0&\leq \lambda - \mu_i\\
    \mu_{i} &\geq  \sum_{\alpha_i^2} \delta^1_i(\alpha_i^1)c_i \\
    1&=\sum_{\alpha_i^1\in \mathcal A_i^1}\delta_i^1(\alpha_i^1) \quad \forall i \\
    0 &\leq \delta_i^1(\alpha_i^1) \quad \forall i, \alpha_i^1
\end{align*}

\subsection{One-sided continuous two-level Blotto game with min aggregator and linear battlefield utilities} \label{dual:one-sided-cont-min}
\paragraph{Max-min strategy} The best response of Player 1 to a fixed strategy $y^2$ of Player 2 can be written as
\begin{align*}
    \min_{y^1}\sum_{i\in[n]} \sum_{\alpha_i^1\in \mathcal A_i^1} \sum_{\alpha_i^2\in \mathcal A_i^2} &y^1(i,\alpha_i^1) y^2_{\contsoldiers_i^2}(i,\alpha_i^2) c_i \\
    \sum_{i\in [n]}y^1(i) &= 1\\
    \sum_{\alpha_i^1\in \mathcal A_i^1}y^1(i,\alpha_i^1) &= y^1(i) \quad \forall i \\
    y^1(i),y^1({i,\alpha_i^1}) &\geq 0 \quad \forall i, \alpha_i^1
\end{align*}
By deriving the dual of the above LP, we obtain 
\begin{align*}
\max_{\lambda,\mu_i} &~\lambda  \\
    0&\geq \lambda - \mu_i\\
    \mu_{i} &\leq  \sum_{\alpha_i^2} y^2_{\contsoldiers_i^2}(i,\alpha_i^2) c_i \quad \forall i,\alpha_i^1
\end{align*}
Finally, a NE of the two-level Blotto game is obtained by solving the linear program
\begin{align*}
\max &~\lambda  \\
    \lambda - \mu_i&\leq 0\\
      \sum_{\alpha_i^2} y^2_{\contsoldiers_i^2}(i,\alpha_i^2) c_i &\geq \mu_{i} \quad \forall i,\alpha_i^1\\
    \sum_{i \in [n]} y^2_{\contsoldiers_i^2}(i) &= m^2 \\
    \sum\limits_{\alpha_i^2 \in \mathcal{A}_i^2} y^2_{\contsoldiers_i^2}(i, \alpha_i^2) &= y_{\contsoldiers_i^2}^2(i) \qquad \forall i\\
    y^2_{\contsoldiers_i^2}(i),y^2_{\contsoldiers_i^2}(i, \alpha_i^2) &\geq 0 \qquad \forall i,\alpha_i^2
\end{align*}
\paragraph{Min-max strategy} The best response of Player 2 to a fixed strategy $y^1$ of Player 1 can be written as 
\begin{align*}
    \max_{y^2}\sum_{i\in[n]} \sum_{\alpha_i^1\in \mathcal A_i^1} \sum_{\alpha_i^2\in \mathcal A_i^2} &y^1(i,\alpha_i^1) y^2_{\contsoldiers_i^2}(i,\alpha_i^2) c_i \\
    \sum_{i \in [n]} y^2_{\contsoldiers_i^2}(i) &= m^2 \\
    \sum\limits_{\alpha_i^2 \in \mathcal{A}_i^2} y^2_{\contsoldiers_i^2}(i, \alpha_i^2) &= y_{\contsoldiers_i^2}^2(i)  \qquad \forall i\\
     y^2_{\contsoldiers_i^2}(i),y^2_{\contsoldiers_i^2}(i, \alpha_i^2) &\geq 0 \qquad \forall i,\alpha_i^2
\end{align*}
We write the dual of the best response LP as follows
\begin{align*}
\min_{\lambda,\mu_i} &~\lambda m^2  \\
    0&\leq \lambda - \mu_i\\
    \mu_{i} &\geq  \sum_{\alpha_i^1} y^1(i,\alpha_i^1)c_i
\end{align*}
Finally, we obtain a NE of the two-level Blotto game by solving the following linear program
\begin{align*}
\min &~\lambda m^2  \\
    \lambda - \mu_i&\geq 0\\
      \sum_{\alpha_i^1} y^1(i,\alpha_i^1)c_i &\leq \mu_{i} \\
       \sum_{i\in [n]}y^1(i) &= 1\\
    \sum_{\alpha_i^1\in \mathcal A_i^1}y^1(i,\alpha_i^1) &= y^1(i) \quad \forall i \\
    y^1(i),y^1({i,\alpha_i^1}) &\geq 0 \quad \forall i, \alpha_i^1
\end{align*}

\section{Online Learning for Discrete Soldiers with Sum-Aggregators}
\label{sec:online-learning-appendix-discrete-sum-agg}
The key idea behind the online-learning based approach is to construct no-regret online learners over the convex space of strategies (typically a polytope) for each player, where the expected utility is convex-concave in this space, and where the extreme points in the polytope correspond to pure strategies. Then, it follows from standard theory that the Nash equilibrium can be approached by adopting \textit{self-play} between the two regret minimizers and averaging the recommended strategies from each online learner. This general approach is responsible for most of the practical state-of-the-art game solvers, particularly imperfect information extensive form games where the convex set of strategies is given by the \textit{treeplex}\cite{von1996efficient} and the sequence form. Our proposed approach is similar at a high level, with the main technical contribution being to describe the space of strategies and constructing regret minimizers for them.

Recall that in the polytope $\Gamma^j$ the variables $h^j_{i-1,a,c}$ give the probability that for player $j$, $m^j-a$ soldiers were placed in battlefields $1$ to $i-1$, \textit{and} $c-a$ soldiers were placed in battlefield $i$. Therefore, $\sum_{k\leq r\leq m^j} h^{j}_{i,r,r-k_i^j}$ is the probability that $k^j_i$ soldiers were placed by player $j$ in battlefield $i$. Furthermore, $x^j_{i,k_i^j,\alpha_i^j}$ is the probability that $k^j_i$ soldiers were placed by player $j$ in the $i$-th battlefield and action $\alpha^j_i$ was chosen as its action on the $i$-th battlefield. 

These constraints are more easily visualized by Figures~\ref{fig:layered_graph} and \ref{fig:battlefield_seq_form}, which extends the prior representation by \citet{behnezhad2023fast}. The strategy space of simply allocating soldiers is given in Figure~\ref{fig:layered_graph}, while the battlefield level strategies are shown in Figure~\ref{fig:battlefield_seq_form}, for a single battlefield and soldiers placed in that battlefield.

\begin{figure}
\begin{tikzpicture}[scale=1, vertex/.style={circle, draw, fill=white, inner sep=1pt}]
  \foreach \x in {1,...,4}{%
    \foreach \y in {1,...,5}{%
      \node[vertex] (v\x\y) at (\x*3,\y) {};%
    }%
  }

  \node[vertex,draw=none] (v05) at (1.3,5){};

  \foreach \x in {2,...,3}{%
    \pgfmathtruncatemacro{\xn}{\x+1}
    \foreach \y in {1,...,5}{%
      \pgfmathtruncatemacro{\ymax}{\y}
      \foreach \yp in {1,...,\ymax}{%
        \draw[->] (v\x\y) -- (v\xn\yp);
      }%
    }%
  }

  \draw[->,color=red] (v25) -- (v33);
  \draw[->,color=red] (v24) -- (v32);
  \draw[->,color=red] (v23) -- (v31);

  \foreach \y in {1,...,5}{%
      \pgfmathtruncatemacro{\ymax}{\y}
        \draw[->] (v15) -- (v2\y);
    }%
 \draw[->] (v05) -- (v15);

 \node (bf1) at (4.5,0.3) {\large \textbf{Battlefield 1}};
 \node (bf2) at (7.5,0.3) {\large \textbf{Battlefield 2}};
 \node (bf3) at (10.5,0.3) {\large \textbf{Battlefield 3}};

 \node (l5) at (0,5) {\large \textbf{\#soldiers = 4}};
 \node (l4) at (0,4) {\large \textbf{\#soldiers = 3}};
 \node (l3) at (0,3) {\large \textbf{\#soldiers = 2}};
 \node (l2) at (0,2) {\large \textbf{\#soldiers = 1}};
 \node (l1) at (0,1) {\large \textbf{\#soldiers = 0}};
\end{tikzpicture}
\caption{Layered graph structure used to compactly describe the mixed strategy space of a single player where there are 3 battlefields and 4 soldiers (thus a height of 5). Note that this edges correspond to the $h^j$ variables and not the battlefield level strategies. Each path from left to right corresponds to a unique discrete distribution of soldiers over battlefields, i.e., a pure strategy. For instance, the path in green shows represents the pure strategy where 3 soldiers are allocated to battlefield 1, 1 to battlefield 2, and none to battlefield 3. Furthermore, we can also see that the marginal distribution of soldiers on each battlefield. For example, the probability that 2 soldiers are placed in battlefield $2$ is the sum of the edges in red.
Note that for simplicity in illustration we allow soldiers to be unused (paths can end at any of the rightmost vertices, not just at the vertex corresponding to $0$ soldiers remaining). This typically does not lead to issues since in most applications, payoffs are monotonically non-decreasing in the number of soldiers. However, it is entirely possible to remove edges from the second-last to last layer to enforce this. Without the battlefield level strategies (see Figure~\ref{fig:battlefield_seq_form}, this is identical to the approach of \citet{behnezhad2023fast}.
}
\label{fig:layered_graph}
\end{figure}
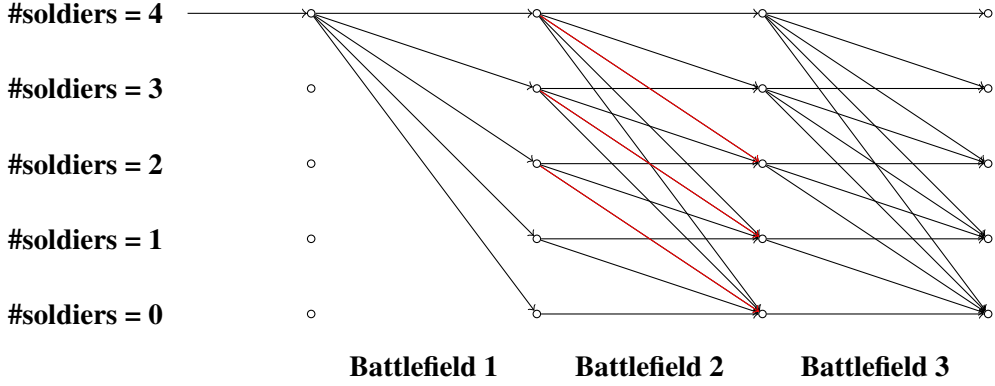
\begin{figure}
\centering
\begin{minipage}{0.49\textwidth}
\centering
\begin{tikzpicture}[scale=1, vertex/.style={circle, draw, fill=white, inner sep=1pt}]
  \foreach \x in {2,...,3}{%
    \foreach \y in {1,...,5}{%
      \node[vertex] (v\x\y) at (\x*3,\y) {};%
    }%
  }

 \node[vertex] (q1) at (2.5*3-0.33, 4.2){=};
 \node[vertex] (q2) at (2.5*3, 3){=};
 \node[vertex] (q3) at (2.5*3+0.33, 1.8){=};

\draw[->,color=red] (v25) -- (q1);
  \draw[->,color=red] (v24) -- (q2);
  \draw[->,color=red] (v23) -- (q3);

  \draw[->,color=red] (q1) -- (v33);
  \draw[->,color=red] (q2) -- (v32);
  \draw[->,color=red] (q3) -- (v31);
 
 \node[vertex] (w) at (2.5*3, 0){+};

\draw[->,color=blue] (q1) -- (w);
\draw[->,color=blue] (q2) -- (w);
\draw[->,color=blue] (q3) -- (w);

\node[vertex,rectangle] (a1) at (2.5*3-1, -1){};
\node[vertex,rectangle] (a2) at (2.5*3+1, -1){};

\draw[->,color=violet] (w) -- (a1);
\draw[->,color=violet] (w) -- (a2);
\end{tikzpicture}
\end{minipage}
\begin{minipage}{0.49\textwidth}
\centering
\begin{tikzpicture}[scale=1, vertex/.style={circle, draw, fill=white, inner sep=1pt}]
  \foreach \x in {2,...,3}{%
    \foreach \y in {1,...,5}{%
      \node[vertex] (v\x\y) at (\x*3,\y) {};%
    }%
  }

 \node[vertex] (q1) at (2.5*3-0.8, 4.75){=};
 \node[vertex] (q2) at (2.5*3 -0.33, 3.6){=};
 \node[vertex] (q3) at (2.5*3+0.33, 2.4){=};
 \node[vertex] (q4) at (2.5*3+0.5, 1.35){=};

\draw[->,color=red] (v25) -- (q1);
  \draw[->,color=red] (v24) -- (q2);
  \draw[->,color=red] (v23) -- (q3);
  \draw[->,color=red] (v22) -- (q4);

  \draw[->,color=red] (q1) -- (v34);
  \draw[->,color=red] (q2) -- (v33);
  \draw[->,color=red] (q3) -- (v32);
  \draw[->,color=red] (q4) -- (v31);
 
 \node[vertex] (w) at (2.5*3, 0){+};

\draw[->,color=blue] (q1) -- (w);
\draw[->,color=blue] (q2) -- (w);
\draw[->,color=blue] (q3) -- (w);
\draw[->,color=blue] (q4) -- (w);

\node[vertex,rectangle] (a1) at (2.5*3-1, -1){};
\node[vertex,rectangle] (a2) at (2.5*3+1, -1){};
\node[vertex,rectangle] (a15) at (2.5*3, -1){};

\draw[->,color=violet] (w) -- (a1);
\draw[->,color=violet] (w) -- (a2);
\draw[->,color=violet] (w) -- (a15);
\end{tikzpicture}
\end{minipage}
\caption{Left: Sequence form representation for battlefield 2 of Figure~\ref{fig:layered_graph} \textit{when there are 2 soldiers placed in it}. Assume there are 2 actions to be taken at battlefield 2. Red edges correspond to the cases where exactly 2 soldiers are allocated to the battlefield (these are the same red edges in Figure~\ref{fig:layered_graph}) For each of them, their values are summed, giving the total probability that battlefield 2 has 2 soldiers (see edges in blue leading to the ``+'' vertex). This total probability at the ``+'' vertex is then split up into the two battlefield actions, given in purple. These edges correspond to the $x^j$'s; in this specific case $x^j_{2,2,\alpha_2^j}$, i.e., the marginal probability that 2 soldiers are placed in battlefield 2 \textit{and} action $\alpha^j_2$ was chosen. Note that the intersection between the blue and red edges do \textit{not} involve splitting the values up like flow, but rather, are duplicated at the points labelled ``=''. We reiterate that this is only for battlefield 2 with two soldiers,; this is repeated for each other battlefield, and for every possible number of soldiers allocated to it. Right: a similar situation but for another battlefield where there is one soldier allocated to it, and with $3$ battlefield level actions for the player.}
\label{fig:battlefield_seq_form}
\end{figure}
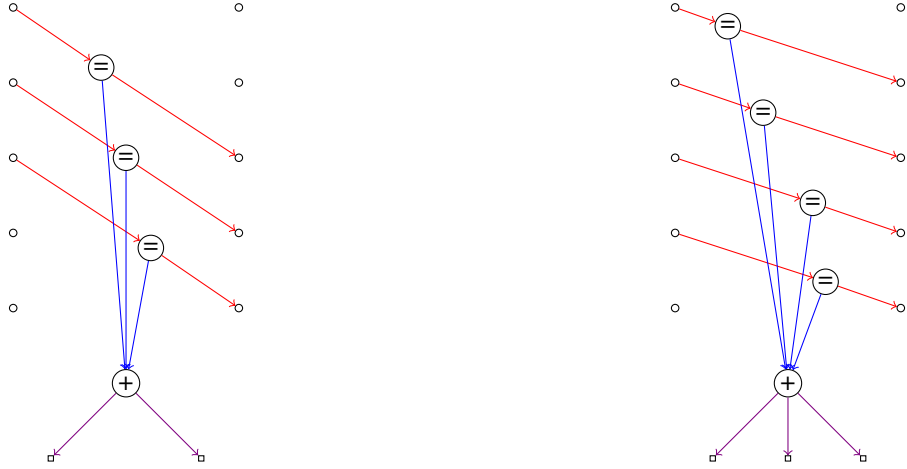

We can see that the strategy space (including mixed strategies) of each player can be represented by a directed acyclic graph (DAG). 
This is almost identical to the usual treeplex \cite{von1996efficient} structure used in representing strategies in extensive-form games, except for the slight different that each ``information set'' could have more than one parent sequence (note that there are no cycles). Given (mixed) strategies for each player, the payoff is \textit{bilinear} in the variables $x^1_{i, \alpha^1_i, k^1_i}$ and $x^2_{i, \alpha^2_i, k^2_i}$. 

What remains is to construct a regret minimizer over this polytope. In this paper we adopt the \textit{scaled extension} approach of \citet{farina2019efficient} due to it's simplicity. An alternative approach is the kernel method of \citet{farina2022kernelized} or \citet{takimoto2003path}. 

The approach of \citet{farina2019efficient} roughly parallels the famous \textit{counterfactual regret minimization} (CFR) approach of \citet{zinkevich2007regret} in that each information set (this is also a vertex or equivalently, decision point in our DAG) contains it's own local regret minimizer (typically over the simplex or some scaled variant). The rewards that the local regret minimizer ``sees'' is what the player would have obtained if it started making decisions at that information set, assuming all its descendants played according to their recommendations. This logic is essentially identical to CFR\footnote{In CFR, each local regret minimizer minimizes their \textit{counterfactual regret}. This is usually thought of as the player ``playing towards'' that particular information set, and is well defined since there is a single sequence leading to that information set (due to perfect recall). In our DAG setting, there are actually multiple sequences leading to most decision points. Hence, we use the equivalent definition of having the player begin at that information set.}. This is done by first performing topological sort the information sets in $\Gamma^j$. Next, we incrementally construct a regret minimizer by adding vertices in this sorted order. Adding a vertex involves appending a new simplex corresponding to the number of actions that exist in that vertex\footnote{We are technically not restricted to adding the simplex, but any other set that we are able to perform regret minimizaiton over. However, for our purposes, the simplex is enough.} This vertex is scaled by some linear combination of it's predecessors in the topological ordering (hence the name \textit{scaled extension}). This construction gives rise to a regret minimizer over $\Gamma^j$ that recursively calls the regret minimization method used at each simplex (e.g., regret matching, regret-matching plus \cite{tammelin2015cfrplus}, predictive regret matching \cite{farina2021predictive} etc.).

\section{Projected Subgradient Ascent Algorithm}\label{app:algo}

\begin{algorithm}[H]
    \caption{Projected Subgradient Ascent}
    \label{alg:psa}
    \begin{algorithmic}[1]
  \Require 
    Set of soldiers assignments \(\Sigma^2\), subgame utilities \(u_i(.)\), initial step size \(\eta>0\), iterations \(T\).
  \Ensure 
    Optimal allocation \(\sigma^{2*}\).
  \State \(\sigma^{2} \gets\) arbitrary point in $\Sigma^2$ \Comment{initial allocation}
  \For{\(t = 1\) \(\mathbf{to}\) \(T\)}
  \For{\(i = 1\) \(\mathbf{to}\) \(n\)} 
      \State \((\delta^{1*}_i,\delta^{2*}_i, v_i^*) \gets \mathrm{SolveSubgame}(G_i(\sigma^2_i))\)
        \Comment{compute Nash strategies and value of \(G_i\)}
      \State 
        \(g_i \gets \displaystyle\sum_{\alpha_i^1,\alpha_i^2}\!
         \delta_i^{1*}(\alpha_i^1)\,\delta_i^{2*}(\alpha_i^2)\,
         \partial_{\sigma_i^2}u_i(\alpha_i^1,\alpha_i^2,\sigma^2_i)\)  
        \Comment{subgradient of \(v_i^*\)}
    \EndFor
    \State \(i^* \gets \arg\min_{i} v_i^*; \quad g \gets 0;\quad g_{i^*}\gets g_{i^*}\)   \Comment{ get descent direction}
    \State \(\eta \gets \mathrm{StepSizeUpdate}(\eta, t), \quad \sigma^2 \gets P_{\Sigma^2}\bigl(\sigma^2 + \eta\,g\bigr)\) \Comment{gradient step + project}
  \EndFor
  \State \Return \(\sigma^2\)
\end{algorithmic}
\end{algorithm}
\section{Proofs}

\subsection{Proof of~\Cref{prop:nash_minimax}}
\nashminimax*
\begin{proof}
We will prove the individual implications separately. 
\begin{itemize}
\item[$\rightarrow$] 
Suppose 
  \[
    V := \max_{x\in X}\inf_{y\in Y}u(x,y)
    \;=\;
    \min_{y\in Y}\sup_{x\in X}u(x,y).
  \]
  By hypothesis there exist 
  \[
    x^* \in \arg\max_{x\in X}\inf_{y\in Y}u(x,y),
    \qquad
    y^* \in \arg\min_{y\in Y}\sup_{x\in X}u(x,y).
  \]
  Even if the minimax theorem doesn't hold, for all $x\in X$, $y\in Y$, we always have the following
  \begin{equation*}
    \inf_{y\in Y}u(x^*,y)\;\le\;u(x^*,y)
    \quad\text{and}\quad
    u(x,y^*)\;\le\;\sup_{x\in X}u(x,y^*).
  \end{equation*}
  Evaluating at $x^*$ and $y^*$ and using the definition of $V$, we obtain
  \[
    V = \inf_{y\in Y}u(x^*,y)
      \;\le\; u(x^*,y^*)
      \;\le\; \sup_{x\in X}u(x,y^*) = V.
  \]
  Hence all three quantities coincide:
  \[
    \inf_{y\in Y}u(x^*,y)
    = u(x^*,y^*)
    = \sup_{x\in X}u(x,y^*) = V.
  \]
  It follows that for every \(x\in X\) and $y\in Y$,
  \[
    u(x,y^*) \;\le\; u(x^*,y^*) \;\le\; u(x^*,y),
  \]
  which is exactly the saddle-point condition.  This means that neither player can improve by unilateral deviation from $(x^*,y^*)$, so $(x^*,y^*)$ is a Nash equilibrium of the zero-sum game with payoff $u$.

Note, however, that this result does not hold if we only have $\sup_x\inf_y u(x,y) = \inf_y\sup_x u(x,y)$. In fact, if the value cannot be actually attained by any existing $(x^*,y^*)$, then this does not imply the existence of the equilibrium. 

\item[$\leftarrow$] Suppose $(x^*,y^*)$ is a Nash equilibrium in the zero‐sum game with payoff $u\colon X\times Y\to\mathbb R$.  By definition of equilibrium, for every $x\in X$ and $y\in Y$, we have
  \[
    u(x,y^*) \;\le\; u(x^*,y^*) \;\le\; u(x^*,y).
  \]
  Taking suprema over $x$ and infima over $y$ yields
  \begin{align*}
    \sup_{x\in X}u(x,y^*) &\le u(x^*,y^*)\le \inf_{y\in Y}u(x^*,y).
  \end{align*}
Hence
  \[
    \inf_{y\in Y}\!\sup_{x\in X}u(x,y)
    \;\le\;
    \sup_{x\in X}u(x,y^*)
    \;\le\;
    u(x^*,y^*)
    \;\le\;
    \inf_{y\in Y}u(x^*,y)
    \;\le\;
    \sup_{x\in X}\!\inf_{y\in Y}u(x,y).
  \]
  On the other hand, by weak duality, 
  \[\sup_x\inf_y u(x,y)\le\inf_y\sup_x u(x,y)\]
  holds for any function $u$.  
  Combining these two chains of inequalities forces
  \[
    \sup_{x\in X}\inf_{y\in Y}u(x,y)
    \;=\;
    \inf_{y\in Y}\sup_{x\in X}u(x,y)
    \;=\;
    u(x^*,y^*).
  \]
  Thus the minimax and maximin values coincide and are attained at $(x^*,y^*)$.  In particular, $\max_{x}\inf_{y}u(x,y)=\min_{y}\sup_{x}u(x,y)$, as claimed.

\end{itemize}
\end{proof}

\subsection{Proof of~\Cref{thm:kuhn-2sided-sum-disc} }
\kuhntwosidedsumdisc*
\begin{proof}
    The treeplex constraints guarantee that
    \begin{align*}
        x^j_{i,k_i^j,\alpha_i^j}&=\gamma^j(\text{assignment with }k_i^j\text{ soldiers placed in battlefield } i) \cdot \delta^j_{i,k_i^j}(\alpha_i^j) \\
        &= \sum_{{k^j\in\mathcal K^j}\atop{k^j \ni k_i^j\text{ fixed}}} \gamma^j(k^j)\cdot \delta_{i,k_i^j}^j(\alpha_i^j).
    \end{align*}
    We proceed by rearranging the formula for computing the expected utility given distributions $\gamma^j\in\Delta(\mathcal K^j)$ over soldiers assignments and $\delta_i^j\in \Delta(\mathcal A_i^j)$ over subgame actions. The overall utility can be written as
    \begin{align*} 
        &\sum_{i\in[n]} \sum_{k^1\in\mathcal K^1}\sum_{k^2\in \mathcal K^2} \gamma^1(k^1)\gamma^2(k^2)  \biggl(\sum_{\alpha_i^1\in \mathcal A_i^1} \sum_{\alpha_i^2\in \mathcal A_i^2}\delta^1_{i,k_i^1}(\alpha_i^1)\delta^2_{i,k_i^2}(\alpha_i^2) u_i(\alpha_i^1,\alpha_i^2,k_i^1,k_i^2) \biggr) \\
        &=\sum_{i\in[n]} \sum_{k_i^1\in [m^1]}\sum_{k_i^2\in [m^2]} \biggl[\sum_{\alpha_i^1,\alpha_i^2} u_i(\alpha,k)  \biggl( \sum_{{k^1\in\mathcal K^1}\atop{k^1\ni k_i^1\text{ fixed}}} \gamma^1(k^1)\cdot \delta_{i,k_i^1}^1(\alpha_i^1)\biggr) \biggl( \sum_{{k^2\in\mathcal K^1}\atop{k^2\ni k_i^2\text{ fixed}}} \gamma^2(k^2)\cdot \delta_{i,k_i^2}^2(\alpha_i^2)\biggr) \biggr] \\
        &=\sum_{i\in[n]} \sum_{\alpha_i^{1} \in \mathcal{A}^{1}_i} \sum_{\alpha_i^{2} \in \mathcal{A}^{2}_i}
        \sum_{k_i^{1}\in[m^{1}]}
        \sum_{k_i^{2}\in[m^{2}]}
        u_i (\alpha_i^{1}, \alpha_i^{2}, k^{1}_i, k^{2}_i) \cdot x^{1}_{i,k_i^{1},\alpha_i^{1}}
    \cdot x^{2}_{i,k_i^{2},\alpha_i^{2}},
    \end{align*}
    where the equality follows form regrouping by battlefield allocations: since $k_i^j$ is the number of soldiers assigned to battlefield $i$ in assignment $k^j$, we can rewrite the sums over full assignments by grouping together those assignments that share the same $k_i^j$.
\end{proof}

\subsection{Proof of~\Cref{prop:comp-2sided-sum-disc}}
\comptwosidedsumdisc*
\begin{proof}
    The representation of $\Gamma^j$ is polynomial in the size of the input. To see this, note that for each battlefield $i$, there are on the order of $(m^j+1)^2$ the flow variables $h^j$. Moreover, for each battlefield $i$ and each possible soldier allocation $k_i^j$ (ranging from 0 to $m^j$), there exists a variable $x^j_{i,k_i^j,\alpha_i^j}$ for each action $\alpha_i^j\in \mathcal A_i^j$. This gives roughly $(m^j+1) a_i^j$ variables $x^j$ per battlefield. Overall, the total number of variables is roughly $ O(n((m^j+1)^2+(m^j+1)a_i^j))$. The number of constraints also scales in a similar polynomial manner in terms of $n,m,i$ and $a_i^j$. Therefore, the polytope $\Gamma^j$ has a polynomial-size representation, and the resulting LP formulation of the two-level Blotto game can be solved in polynomial time.
\end{proof}

\subsection{Proof of~\Cref{thm:non-exist-2sided-min-disc}}
\existtwosidedmindisc*
\begin{proof}
Consider a two-sided discrete two-level Blotto game with two battlefields. 
In each battlefield, both players have a single action and share the same utility function defined in terms of the (discrete) soldiers assignments of the players:
\begin{equation*}
    u_i(\emptyset, \emptyset, k^{1}_i, k^{2}_i) = \frac{k^1_i+1}{k^2_i+1}, \quad i\in\{1,2\}.
\end{equation*}
Player 1 seeks to maximize their utility, while Player 2 aims to minimize it, making the utility function increasing in each player's respective allocation. Suppose that both players have 2 soldiers, i.e. $k^j_i\in\{0,1,2\}, \,i, j\in\{1,2\}$.

When Players 1 and 2 use a mixed strategy $\gamma^1$ and $\gamma^2$ respectively over soldiers assignments, the expected utility on battlefield $i$ is \[
\mathbb E[u_i]=\sum_{k^1\in \mathcal K^1}\sum_{k^2\in \mathcal K^2} \gamma^1(k^1)\gamma^2(k^2)u_i(k^1_i,k^2_i).
\]
The overall utility is then
\[
U=\min\bigl\{\mathbb E[u_1(k^1_1,k^2_1)],\,\mathbb E[u_2(k^1_2,k^2_2)] \bigr\}.
\]
We show that the max-min and min-max strategies differ, and therefore no Nash equilibrium exists in this game.

\paragraph{Min-max and max-min values} We now compute the min-max and max-min values of the game. For simplicity, we let $k^j=k^j_1$ be the number of soldiers Player $j$ allocates to battlefield 1, which implies that $2-k^j$ soldiers are Player $j$'s soldiers on battlefield 2.

\textit{Min-max value.}  Let $\gamma^2$ be arbitrary, and denote \[
t_{\ell}=\gamma^2({\ell})=\mathbb P[k^2={\ell}], \qquad {\ell}\in\{0,1,2\},\, \sum_{{\ell}=0}^2 t_{\ell}=1.
\]
We compute payoffs for every allocation $k^1\in\{0,1,2\}$ of Player 1 to battlefield 1. For each $k^1$, define $U_t(k^1)=\min\{\mathbb E[u_1],\,\mathbb E[u_2]\}$, where 
\begin{align*}
&\mathbb E[u_1] = \sum_{{\ell}=0}^2 t_{\ell}\, u_1(k^1,{\ell}) = \sum_{\ell=0}^2 \, t_{\ell} \,\frac{k^1+1}{\ell+1} = (k^1+1) \,\alpha, \quad \text{where } \alpha = t_0 + \tfrac12 t_1 + \tfrac13 t_2.\\
&\mathbb E[u_2] = \sum_{\ell=0}^2 t_{\ell}\,u_2(2-k^1,2-\ell) = \sum_{\ell=0}^2 t_{\ell}\,\frac{3-k^1}{3-\ell} = (3-k^1) \,\beta, \quad \text{where } \beta = \tfrac13 t_0 + \tfrac12 t_1 + t_2.
\end{align*}
One checks that:\[
U_t(0) = \min\{\alpha, 3\beta\}, \qquad U_t(1)=\min\{2\alpha, 2\beta\}, \qquad U_t(2) = \min\{3\alpha,\beta\}
\]
Player 1's best response value is $\phi(t_0,t_1,t_2) = \max_{k^1\in\{0,1,2\}} U_t(k^1)$. This function is minimized at $(t_0,t_1,t_2)=(0.6,0.4,0)$, which gives $\min \phi(t_0,t_1,t_2) = 0.8$, implying that the min-max value is 0.8.

\textit{Max-min value.} Let $\gamma^1$ be arbitrary, and denote \[
p_{\ell}=\gamma^1({\ell})=\mathbb P[k^1={\ell}], \qquad {\ell}\in\{0,1,2\},\, \sum_{{\ell}=0}^2 t_{\ell}=1.
\]
We compute payoffs for every allocation $k^2\in\{0,1,2\}$ of Player 2 to battlefield 1. For each $k^2$, define $U_t(k^2)=\min\{\mathbb E[u_1],\,\mathbb E[u_2]\}$, where 
\begin{align*}
&\mathbb E[u_1] = \sum_{{\ell}=0}^2 p_{\ell}\, u_1({\ell},k^2) = \sum_{\ell=0}^2 \, p_{\ell} \, \frac{\ell+1}{k^2+1} = \frac{\alpha}{k^2+1}, \quad \text{where } \alpha = p_0 + 2 p_1 + 3 p_2.\\
&\mathbb E[u_2] = \sum_{\ell=0}^2 p_{\ell}\,u_2(2-\ell,2-k^2) = \sum_{\ell=0}^2 p_{\ell}\,\frac{3-\ell}{3-k^2} = \frac{\beta}{3-k^2}, \quad \text{where } \beta = 3 p_0 + 2 p_1 + p_2.
\end{align*}
One checks that:\[
U_t(0) = \min\{\alpha, \tfrac13\beta\}, \qquad U_t(1)=\min\{\tfrac12\alpha, \tfrac12\beta\}, \qquad U_t(2) = \min\{\tfrac13\alpha,\beta\}
\]
Player 2's best response value is $\phi(p_0,p_1,p_2) = \min_{k^2\in\{0,1,2\}} U_t(k^2)$. This function is maximized at $(p_0,p_1,p_2) = (\tfrac12,0,\tfrac12)$, which gives $\max \phi(p_0,p_1,p_2) = \tfrac23$, implying that the max-min value is $\tfrac23$.
\end{proof}

\subsection{Proof of~\Cref{thm:kuhn-one-sided-disc-min}}
\kuhnonesidedmindisc* 
\begin{proof}    
    We rearrange the formula for computing the expected utility given distributions $\gamma^2\in\Delta(\mathcal K^2)$ over soldiers assignments for Player 2 and $\delta_i^j\in \Delta(\mathcal A_i^j), j\in\{1,2\}$ over subgame actions for Players 1 and 2.  
    \begin{align}
        U&=\min_{i\in[n]} \sum_{k^2\in\mathcal K^2} \gamma^2(k^2) \mathbb E[u_i(\delta^1_i,\delta^2_{i,k_i^2},k_i^2)] \\
        &=\min_{p\in\Delta(n)} \sum_{i\in [n]} p_i \biggl(\sum_{k^2\in\mathcal K^2} \gamma^2(k^2) \mathbb E[u_i(\delta_i^1,\delta_{i,k_i^2}^2,k_i^2)]\biggr)\label{eq2:one-sided-min-corr-thm} \\
        &= \min_{p\in\Delta(n)} \sum_{i\in[n]} \sum_{k_i^2 \in [m^2]} \sum_{\alpha_i^2\in\mathcal A_i^2} p_i \mathbb E_{\alpha_i^1 \sim \delta^1_i} [u_i(\alpha_i^1,\alpha_i^2,k_i^2)] \biggl( \sum_{{k^2\in\mathcal K^2}\atop{k^2\ni k_i^2\text{ fixed}}} \gamma^2(k^2)\cdot \delta^2_{i,k_i^2}(\alpha_i^2)\biggr)\label{eq3:one-sided-min-corr-thm} \\
        &= \min_{p\in\Delta(n)} \sum_{i\in[n]} \sum_{\alpha_i^2\in \mathcal A_i^2} \sum_{k_i^2 \in [m^2]} p_i \mathbb E_{\alpha_i^1 \sim \delta^1_i} [u_i(\alpha_i^1,\alpha_i^2,k_i^2)] x^2_{i,k_i^2,\alpha_i^2},
    \end{align} 
    where equality \eqref{eq2:one-sided-min-corr-thm} follows from convexifying the minimum, and equality \eqref{eq3:one-sided-min-corr-thm} follows from regrouping by battlefield allocations for Player 2. 
\end{proof}

\subsection{Proof of~\Cref{thm:sion-one-sided-disc-min}}
\siononesidedmindisc*
\begin{proof}
    Consider the utility formulation from~\Cref{thm:kuhn-one-sided-disc-min}. Then the min-max problem can be written as 
    \begin{align}
    \min_{\delta^{1}}\max_{\delta^2,\gamma^2}&\min_{i\in[n]} \mathbb E[u_i(\alpha_i^1,\alpha_i^2,k_i^2)] \label{eq:thm5-min-max} \\
    &=\min_{\delta^{1}}\max_{x^2}\min_{p\in \Delta(n)}  \sum_{i\in [n]}\sum_{\alpha_i^2\in \mathcal A_i^{2}}\sum_{k_i^2\in [m^{2}]}p_i x^2_{i,k_i^2,\alpha_i^2} \mathbb E_{\alpha_i^1 \sim \delta^1_i}[u_i(\alpha_i^{1},\alpha_i^2,k_i^2)]. 
    \end{align}
    Since the objective is bilinear in $p,x^2$ which are defined over convex and compact sets, we can exchange the inner $\max$ and $\min$ by Sion's minimax theorem. We can now merge the outer two minimizations by defining the standard sequence form over a two-level game where Player 1 first chooses a battlefield and then plays an action in that battlefield, using $y^1({i,\alpha^1_i})$ to denote the product $p_i\delta_i^{1}(\alpha^1_i)$. Therefore, we get:
    \begin{equation*}
    \begin{aligned}
    &\min_{\delta^{1}}\max_{x^2}\min_{p\in \Delta(n)}  \sum_{i\in [n]}\sum_{\alpha_i^2\in \mathcal A_i^{2}}\sum_{k_i^2\in [m^{2}]}p_i x^2_{i,k_i^2,\alpha_i^2} \mathbb E_{\alpha_i^1 \sim \delta^1_i}[u_i(\alpha_i^{1},\alpha_i^2,k_i^2)] \\ &=
    \min_{\delta^1,p}\max_{x^2}\sum_{i\in [n]}\sum_{\alpha^{1}_i\in \mathcal A_i^{1}}\sum_{\alpha^{2}_i\in \mathcal A_i^{2}}\sum_{k_i^{2}\in [m^{2}]} p_i \delta^1_i(\alpha_i^1) x^2_{i,k_i^{2},\alpha_i^{2}} u_i(\alpha_i^{1}, \alpha_i^{2},k_i^{2})
    \\ &= \min_{y^1}\max_{x^2}\sum_{i\in [n]}\sum_{\alpha^{1}_i\in \mathcal A_i^{1}}\sum_{\alpha^{2}_i\in \mathcal A_i^{2}}\sum_{k_i^{2}\in [m^{2}]}y^1({i,\alpha_i^{1}}) x^2_{i,k_i^{2},\alpha_i^{2}}u_i(\alpha_i^{1}, \alpha_i^{2},k_i^{2}).
\end{aligned}
\end{equation*}
Since the objective is now bilinear, we apply Sion's minimax theorem to exchange the outer $\min$ and $\max$. Next, we disaggregate $y^1$ into $p,\delta^{1}$ and $x^2$ into $\delta^{2},\gamma^2$. This yields the following formulation of the min-max formulation in \eqref{eq:thm5-min-max}:
\begin{align*}
  & \max_{\delta^{2}} \max_{\gamma^2}\min_{\delta^{1}}\min_{p\in \Delta(n)} \sum_{i\in[n]} \sum_{k^2\in\mathcal K^2} \gamma^2(k^2) \biggl(\sum_{\alpha_i^1\in\mathcal A_i^1}\sum_{\alpha_i^2\in\mathcal A_i^2} p_i\delta_i^1(\alpha_i^1)\delta_{i,k_i^2}^2(\alpha_i^2) u_i(\alpha_i^1,\alpha_i^2,k_i^2) \biggr)   \\ &=\max_{\delta^2,\gamma^2}\min_{\delta^{1}}\min_{i\in [n]} \mathbb E[u_i(\alpha_i^{1},\alpha_i^{2},k_i^2)],
\end{align*}
which completes the proof.
\end{proof}

\subsection{Proof of~\Cref{cor:comp-one-sided-min-disc}}
\componesidedmindisc*
\begin{proof}
    The existence of a Nash equilibrium follows from the minimax theorem. 
    The representation of $\mathcal P$ is polynomial in the size of the input: there are $1+n+\sum_{i=1}^n|\mathcal A_i^1|$, and a total of $2+n+\sum_{i=1}^n\mathcal A_i^1$ constraints.  
    Since both polytopes $\Gamma^2$ and $\mathcal P$ have polynomial-size representations, and since LPs can be solved in polynomial time, then computing a Nash equilibrium of this two-level Blotto game can be done in polynomial time.
\end{proof}

\subsection{Proof of~\Cref{thm:exist-two-sided-cont}}
\existtwosidedcont*
\begin{proof}
    Consider two battlefields, each featuring a trivial game where both players have a single action and share the same utility function which we define as a function of the (continuous) number of soldiers $\contsoldiers^{1},\contsoldiers^{2}$ assigned by Players 1 and 2 respectively:
\begin{equation*}
    u(\emptyset, \emptyset, \contsoldiers^{1}, \contsoldiers^{2}) = \begin{cases}
        \contsoldiers^{1} - \contsoldiers^{2} & \text{if}~\contsoldiers^{1} \ge \contsoldiers^{2}\\
        0 & \text{otherwise}.
    \end{cases}
\end{equation*}
The first player seeks to maximize their utility, while the second player aims to minimize it, making the utility function increasing in each player's respective allocation $\contsoldiers^{j}, j\in\{1,2\}$.
Suppose the first player has 2 soldiers and the second player has 1 soldier. We analyze the existence of an equilibrium under the two aggregation methods \textit{(note that the ``otherwise'' condition in the utility definition is essential for the validity of the analysis).} 

\paragraph{Sum-aggregated continuous two-level Blotto game}

We show that in the continuous two-level Blotto game with sum aggregator, the min-max and max-min values of the game differ. As a result, no pure-mixed Nash equilibrium exists.

\begin{itemize}
\item \textit{Max-min value.} 
Fix any pure assignment $\contsoldiers^1\in \Sigma^1 = (a,2-a)$. Player 2 can place his 1 soldier entirely on the battlefield $i\in\{1,2\}$ where $\contsoldiers_i^1$ is largest: 
\begin{itemize}
    \item Play (1,0) when $a \ge 2-a$. In this case $a \ge 1$, which implies $u_1=a-1$ and $u_2=2-a$.
    \item Play (0,1) when $a < 2-a$. In this case $a < 1$, which implies that $u_1=a$ and $u_2=1-a$.
\end{itemize}
Thus for every pure $\contsoldiers^1$, the worst‐case payoff is 1, so $\max_{\sigma^1 \in \Sigma^1} \min_{\sigma^2\in\Sigma^2} u(\contsoldiers^1,\contsoldiers^2) = 1$. 

\item \textit{Min-max value.} 
Fix any pure $\contsoldiers^2\in \Sigma^2$. Player 1’s best response is to place all 2 soldiers on the battlefield where $\contsoldiers^2_i$ is smallest, giving \[
\max_{\contsoldiers^1\in \Sigma^1} u(\contsoldiers^1,\contsoldiers^2) = 2-\min(\contsoldiers^2_1,\contsoldiers^2_2). 
\]
Player 2 then chooses $\contsoldiers^1$ to maximize $\min(\contsoldiers^2_1,\contsoldiers^2_2)$ under $\contsoldiers^2_1+\contsoldiers^2_2=1$. This is achieved by $\contsoldiers^2_1=\contsoldiers^2_2=1/2$. Hence, $\min_{\contsoldiers^2\in\Sigma^2} \max_{\contsoldiers^1 \in \Sigma^1}u(\contsoldiers^1,\contsoldiers^2)=2-1/2=3/2$.
\end{itemize}

\paragraph{Min-aggregated continuous two-level Blotto game} Similarly, we show that the max-min and min-max values of the continuous two-level Blotto game also differ under the min aggregator, which precludes the existence of a PPNE.  

\begin{itemize}
    \item \textit{Max-min value.}
No matter how Player 1 places his 2 soldiers, one of the two fields has at most 1 soldier. Then Player 2 allocates his soldier entirely to that weaker battlefield, driving its utility to zero. Hence, for every pure $\contsoldiers^1$, the worst‐case payoff is 0, so the max-min value of the game is 0.

\item \textit{Min-max value.} 
Fix an allocation $\contsoldiers^2\in \Sigma^2$. Player 1 wants to choose $(\contsoldiers^1_1,\contsoldiers^1_2)$ to maximize $\min\{\contsoldiers^1_1-\contsoldiers^2_1, \, \contsoldiers^1_2-\contsoldiers^2_2\}$. The optimum is given by equalizing the two payoffs, subject to  $\contsoldiers^1_1+\contsoldiers^1_2=2$. One finds the optimum $\contsoldiers^1_1= 1+\frac{\contsoldiers^2_1-\contsoldiers^2_2}{2}$ and $\contsoldiers^1_2= 1-\frac{\contsoldiers^2_1-\contsoldiers^2_2}{2}$, which yields $\max_{\contsoldiers^1 \in \Sigma^1} \min u =1/2$ for every choice of $\contsoldiers^2$. It follows that the min-max value of the game is $1/2$.
\end{itemize}
\end{proof}

\subsection{Proof of~\Cref{prop:exist-one-sided-sum-cont}}
\existonesidedsumcont*
\begin{proof}
    Consider a one-sided continuous two-level Blotto game where Player 1 has two (continuous) soldiers that she assigns to battlefields 1 and 2. We let \( \contsoldiers_1 \in [0,2] \) be the allocation to Battlefield 1, and \( \contsoldiers_2 = 2 - \contsoldiers_1 \) to Battlefield 2.  Only Battlefield 1 features a $(2\times 2)$ normal form subgame $G_1$ with payoff matrix \[
    \begin{bmatrix}
    \contsoldiers_1^2 & 2 \\
    0 & 0
    \end{bmatrix}
    \]
    where Player 1 plays Top with probability $x$ and Bottom with probability $1-x$, and Player 2 plays Left with probability $y$ and Right with probability $1-y$. 
    We assume that Player 1 is the maximizer and Player 2 is the minimizer. Given the payoff matrix, the utility from Battlefield 1 can be written as \[u_1(x, y, \contsoldiers_1) = x (y \contsoldiers_1^2 + 2(1 - y)),\]
    and that from Battlefield 2 is given by \[
    u_2(\contsoldiers_2) = \contsoldiers_2 = 2-\contsoldiers_1.
    \]
    Under the sum aggregator, the total utility is
    \[
    U =  u_1(x, y, \contsoldiers_1) + u_2(\contsoldiers_2).
    \]
We show that a Nash equilibrium of this game fails to exist. 

First, we show that Player 1 always plays Top. By simply looking at the payoffs, we can see that for Player 1, playing Top weakly dominates playing Bottom. 
Moreover, Player 1 would play Bottom with non-zero probability only if Player 2 plays Left, i.e. $y=1$, and $\contsoldiers_1=0$. However, for the case where $y=1$, Player 1's best response would be to play $\contsoldiers_1>1$ and $x=1$. Otherwise, she would achieve a strictly lower utility. Hence, Player 1 never plays Bottom in the equilibrium and so $x=1$. 
    
Next, we analyze the best response of Player 2 to any fixed strategy of Player 1. Let us  fix any \( \contsoldiers_1 \in [0,2] \). Let $A(y) = y \contsoldiers_1^2 + 2(1 - y)$, so that $u_1(x, y, \contsoldiers_1) = x A(y)$. Since Player 1 will choose \( x = 1 \) at equilibrium, it follows that Player 2's goal is to minimize $u = A(y) + u_2(\contsoldiers_1)$. Notice that
\begin{itemize}
    \item if \( \contsoldiers_1^2 > 2 \), then \( A(y) \) is increasing in \( y \), so Player 2 plays \( y = 0 \)
    \item if \( \contsoldiers_1^2 < 2 \), then \( A(y) \) is decreasing in \( y \), so Player 2 plays \( y = 1 \)
    \item if \( \contsoldiers_1^2 = 2 \), Player 2 is indifferent, so $y\in[0,1]$. 
\end{itemize}

We now argue that there is no Player 2 strategy to which $\contsoldiers_1^2=2$ would be a best response to.
Suppose, for contradiction, that there exists a strategy $y \in [0,1]$ for Player 2 such that $\contsoldiers_1 = \sqrt{2}$ is a best response for Player 1. Then $\contsoldiers_1$ must maximize the function
\[
u(\contsoldiers_1, y) = y(\contsoldiers_1^2 - 2) + 4 - \contsoldiers_1
\]
over the interval $\contsoldiers_1 \in [0,2]$. We differentiate with respect to $\contsoldiers_1: \frac{d}{d\contsoldiers_1} u = 2y \contsoldiers_1 - 1:=d(\contsoldiers_1,y)$. At $\contsoldiers_1=\sqrt 2$, notice that $d(y)$ is increasing if $y > \frac{1}{2\sqrt{2}}$ and decreasing if $y < \frac{1}{2\sqrt{2}}$. When $y = \frac{1}{2\sqrt{2}}$, $\frac{d^2}{d\contsoldiers_1^2} u = \frac{1}{\sqrt 2} > 0$ and $\contsoldiers_1 = \sqrt{2}$ is a local minimum.
Therefore, for no $y \in [0,1]$ is $\contsoldiers_1 = \sqrt{2}$ a best response for Player 1. It follows that Player 2 always plays a pure strategy, i.e. $y\in\{0,1\}$.

We now analyze Player 1's best response to pure strategies of Player 2. If  $y = 0$, then Player 1's utility is strictly decreasing in $\contsoldiers_1$, so his best response is $\contsoldiers_1 = 0$, which satisfies $\contsoldiers_1^2 < 2$. If $y = 1$, then Player 1's utility is increasing in $\contsoldiers_1$ on $[1/2,2]$. Hence, Player 1 will always choose $\contsoldiers_1=2$, which satisfies $\contsoldiers_1^2 > 2$, as a best response to $y=1$, since this maximizes his utility on the feasible interval $[0,2]$.
These best responses induce the following cycle:
\[
(y = 1,~ \contsoldiers_1 > \sqrt{2}) \to 
(y = 0,~ \contsoldiers_1 < \sqrt{2}) \to 
(y = 1,~ \contsoldiers_1 > \sqrt{2}) \to \cdots
\]
Thus, no fixed point is reached, and no mutual best response pair exists. It follows that no Nash equilibrium exists in this game.

\paragraph{Max-min and min-max values}
We now compute the max-min and min-max values of the game.

\textit{Max-min value.}
The max-min value is defined as:
\[
\max_{\contsoldiers_1 \in [0,2], x\in[0,1]} \min_{y \in [0,1]} \left[ u_1(1, y, \contsoldiers_1) + u_2(2 - \contsoldiers_1) \right]
\]
Since Player 1 always plays $x = 1$, the expression simplifies to \[\max_{\contsoldiers_1 \in [0,2]} \min_{y \in [0,1]} \left[ y (\contsoldiers_1^2 - 2) + 4 - \contsoldiers_1 \right].\]
When Player 2 is best responding, Player 1 solves the following maximization problem over $\contsoldiers_1$: 
\[
\max_{\contsoldiers_1 \in [0,2]}
\begin{cases}
\contsoldiers_1^2 + 2 - \contsoldiers_1 & \text{if } \contsoldiers_1^2 < 2 \\
4 - \contsoldiers_1 & \text{if } \contsoldiers_1^2 \geq 2
\end{cases}
\]
It follows that the max-min value is $4 - \sqrt{2}$, achieved at $\contsoldiers_1 = \sqrt{2}$ and $x = 1$.

\textit{Min-max value.} The minmax value is defined as:
\[
\min_{y \in [0,1]} \max_{\contsoldiers_1 \in [0,2]} \left[ y (\contsoldiers_1^2 - 2) + 4 - \contsoldiers_1 \right]
\]
When Player 1 is best responding, Player 2 solves the following minimization problem over $y$:
\[
\min_{y \in [0,1]}\biggl(\max_{x,\contsoldiers_1} xy\contsoldiers_1^2-2y-\contsoldiers_1+4\biggr)
\]
The minimum value is 3, attained when $y=1/2$. Therefore, the min-max value is 3.
\end{proof}

\subsection{Proof of~\Cref{thm:sion-one-sided-cont-sum}}
\siononesidedsumcontlinear*
\begin{proof}
We can write the min-max problem as:
\begin{align}
    &\max_{\delta^{2},\sigma^2}\min_{\delta^{1}} \sum_{i\in[n]} \mathbb{E}[u_i(\alpha_i^{1},\alpha_i^{2},\contsoldiers_i^2)] \\ &= \max_{\delta^2,\contsoldiers^2}\min_{\delta^{1}} \sum_{i\in[n]} \biggl(\sum_{\alpha_i^1\in \mathcal A_i^1} \sum_{\alpha_i^2\in \mathcal A_i^2} \delta_{i}^1(\alpha_i^1)\delta^2_{i,\contsoldiers_i^2}(\alpha_i^2) c_i\cdot\contsoldiers_i^2 \biggr) \label{eq2:cont-sum-one-sided} \\
    &= \max_{y^2}\min_{\delta^1}\sum_{i\in[n]} \sum_{\alpha_i^1\in \mathcal A_i^1} \sum_{\alpha_i^2\in \mathcal A_i^2} \delta^1_i(\alpha_i^1)y^2_{\contsoldiers_i^2}(i,\alpha_i^2) c_i \label{eq3:cont-sum-one-sided} \\
    &= \min_{\delta^{1}} \max_{y^2} \sum_{i\in[n]} \sum_{\alpha_i^1\in \mathcal A_i^1} \sum_{\alpha_i^2\in \mathcal A_i^2} \delta^1_i(\alpha_i^1)y^2_{\contsoldiers_i^2}(i,\alpha_i^2) c_i \label{eq4:cont-sum-one-sided} \\
    &= \min_{\delta^{1}} \max_{\delta^2,\contsoldiers^2} \sum_{i\in[n]} \sum_{\alpha_i^1\in \mathcal A_i^1} \sum_{\alpha_i^2\in \mathcal A_i^2} \delta^1_i(\alpha_i^1) \delta^2_i(\alpha_i^2)\contsoldiers^2_i c_i\label{eq5:cont-sum-one-sided} \\
    &= \min_{\delta^{1}} \max_{\delta^2,\contsoldiers^2} \sum_{i\in[n]}\mathbb{E}[u_i(\alpha_i^{1},\alpha_i^{2},\contsoldiers_i^2)]   
\end{align}
where equation \eqref{eq3:cont-sum-one-sided} follows from merging the variables $\contsoldiers^2,\delta^2$ into a sequence-form variable $y^2_{\contsoldiers^2}$ defined over the flow polytope $\mathcal Q$ (see~\Cref{sect:cont-sum} for more details).
Since the objective function is now  bilinear in $y^2$ and $\delta^1$, we can exchange the inner $\min$ and $\max$ using Sion's minimax theorem (eq. \eqref{eq4:cont-sum-one-sided}). Finally, equation \eqref{eq5:cont-sum-one-sided} follows from and  dissagregating $y^2$ back into variables $\contsoldiers^2$ and $\delta^2$.
This completes the proof.
\end{proof}

\subsection{Proof of~\Cref{thm:exist-one-sided-min-cont}}
\existonesidedmindisc*
\begin{proof}
    Consider a one-sided continuous two-level Blotto game where Player 1 has two (continuous) soldiers that he assigns to battlefields 1 and 2. We let \( \contsoldiers_1 \in [0,2] \) be the allocation of Player 1 to Battlefield 1, and \( \contsoldiers_2 = 2 - \contsoldiers_1 \) to Battlefield 2.  Only Battlefield 1 features a $(2\times 2)$ normal form subgame $G_1$ with payoff matrix \[
    \begin{bmatrix}
    1\cdot \textbf{1}[\contsoldiers_1<1] & 0 \\
    0 & (\contsoldiers_1)^2\cdot \textbf{1}[\contsoldiers_1\geq 1]
    \end{bmatrix}.
    \]
    Player 1 plays Top with probability $x$ and Bottom with probability $1-x$, and Player 2 plays Left  with probability $y$ and Right  with probability $1-y$. 
    We assume that Player 1 is the maximizer and Player 2 is the minimizer. Given the payoff matrix, the utility from Battlefield 1 can be written as \[
    u_1(\contsoldiers_1,x,y) =
    \begin{cases}
    x \cdot y &  \contsoldiers_1 < 1 \\
    (\contsoldiers_1)^2\cdot (1-x)(1-y) & \contsoldiers_1 \geq 1,
    \end{cases}
    \]
    and that from Battlefield 2 is given by \[
    u_2(\contsoldiers_1) = 2-\contsoldiers_1.
    \]
    Under the min aggregator, the total utility is
    \[
   U(\contsoldiers_1, x, y) = \min \bigl\{u_1(\contsoldiers_1, x, y), u_2(\contsoldiers_1)\bigr\}.
    \]
We show that a PPNE of this game fails to exist. We first analyze Player 2's best response strategy. Note that for any fixed $(\contsoldiers_1,x)$, we have
\[
U(\contsoldiers_1,x,y)
=
\begin{cases}
\min\{x\,y,\;2-\contsoldiers_1\} = x\,y, & \contsoldiers_1<1 \\
\min\{\contsoldiers_1^2(1-x)(1-y),\;2-\contsoldiers_1\}=\contsoldiers_1^2(1-x)(1-y), & \contsoldiers_1\ge1.
\end{cases}
\]
 If $\contsoldiers_1<1$, $U(\contsoldiers_1,x,y)=x\,y$ is strictly increasing in $y$, so Player 2’s unique minimizer is 
    \[
    y^*(\contsoldiers_1,x)=0.
    \]
If $\contsoldiers_1\ge1$, $U(\contsoldiers_1,x,y)=\contsoldiers_1^2(1-x)(1-y)$ is strictly decreasing in $y$, so Player 2’s unique minimizer is 
    \[
    y^*(\contsoldiers_1,x)=1.
    \]
Hence at any Nash equilibrium we must have
\[
y^* =
\begin{cases}
0,&\contsoldiers_1^*<1,\\
1,&\contsoldiers_1^*\ge1.
\end{cases}
\]
Assume now that $(\contsoldiers_1^*,x^*,y^*)$ is a Nash equilibrium.  We consider the two cases forced by Player 2’s best response.
\begin{itemize}
    \item \textit{Case 1:} $\contsoldiers_1^*<1$.  Then $y^*=0$, and $U(\contsoldiers_1^*,x^*,0)=\min\{x^*\cdot 0,\;2-\contsoldiers_1^*\}=0.$
    But if Player 1 deviates to $(\tilde \contsoldiers_1 = 1, \tilde x = 0),$ then $u_1(1,0,0) =1$ and $u_2(1)=1$, implying that $U(1,0,0)=1>0.$
    Thus Player 1 strictly improves, contradicting that $(\contsoldiers_1^*,x^*,0)$ is a Nash equilibrium.

    \item \textit{Case 2:} $\contsoldiers_1^*\ge1$.  Then $y^*=1$, and $U(\contsoldiers_1^*,x^*,1) =\min\{\contsoldiers_1^{*2}(1-x^*)\cdot 0,2-\contsoldiers_1^*\}=0.$
    However, if Player 1 deviates to $(\hat \contsoldiers_1 = 0,\hat x = 1)$, then $u_1(0,1,1)=1$, and $u_2(0)=2$, implying that $U(0,1,1)=\min\{1,2\}=1>0$.
    Again, Player 1 can make a strict improvement by deviating, which contradicts the NE assumption. 
\end{itemize}

    In both cases, Player 1 can make a profitable deviation.  Therefore no triple \((\contsoldiers_1^*,x^*,y^*)\) can be a Nash equilibrium.
    
\paragraph{Max-min and min-max values}

We now compute the max-min and min-max values of the game.

\textit{Max-min value.}
The max-min value is defined as:
\[
\max_{\contsoldiers_1 \in [0,2], x\in[0,1]} \min_{y \in [0,1]} U(\contsoldiers_1,x,y)
\]

When Player 2 is best responding, we obtain for any fixed \((\contsoldiers_1,x)\)
\[
\min_{y}U(\contsoldiers_1,x,y)
=\begin{cases}
\min_{y}\{\,x\,y,\;2-\contsoldiers_1\}= 0,&\contsoldiers_1<1,\\
\min_{y}\{\,\contsoldiers_1^2(1-x)(1-y),\;2-\contsoldiers_1\}
=0,&\contsoldiers_1\ge1.
\end{cases}
\]
Hence $\max_{\contsoldiers_1,x}\min_{y}U(\contsoldiers_1,x,y) = 0$, and therefore the max-min value 0.

\textit{Min-max value.} The min-max value is defined as:
\[
\min_{y \in [0,1]} \max_{\contsoldiers_1 \in [0,2]} U(\contsoldiers_1,x,y)
\]
Fix a strategy \(y\) of Player 2.  Player 1’s inner maximization splits into two regions:
\begin{itemize}
    \item If \(\contsoldiers_1<1\), then \(u_1(\contsoldiers_1,x,y)=x\,y\) and \(u_2=2-\contsoldiers_1>1\), implying that $U(\contsoldiers_1,x,y)=\min\{x\,y,2-\contsoldiers_1\}=x\,y,$ whose maximum over \(x\) is achieved at \(x=1\), i.e. $\max_x x\cdot y=y$ for any $\contsoldiers_1<1$.

    \item If \(\contsoldiers_1\ge1\), then \(u_1(\contsoldiers_1,x,y) = \contsoldiers_1^2(1-x)(1-y)\) and \(u_2=2-\contsoldiers_1\le1\).  The maximum over \(x\) occurs at \(x=0\), giving $U(\contsoldiers_1,0,y)
   = \min\{\contsoldiers_1^2(1-y),\,2-\contsoldiers_1\}.$
   To maximize this in \(\contsoldiers_1\ge1\) we solve
   \[
     \contsoldiers_1^2(1-y) = 2-\contsoldiers_1
     \quad\Longrightarrow\quad
     (1-y)\,\contsoldiers_1^2 + \contsoldiers_1 -2 =0,
   \]
   whose unique root in \([1,2]\) is
   \[
     \contsoldiers_1^*(y)
     =
     \frac{-1 + \sqrt{\,9-8y\,}}{2\,(1-y)}.
   \]
   At that \(\contsoldiers_1^*(y)\), the common value is $2 - \contsoldiers_1^*(y)$.
\end{itemize}
Putting the two regions together, 
$\max_{\contsoldiers_1,x}U(\contsoldiers_1,x,y)
=\max\Bigl\{\,y,\;2 - \contsoldiers_1^*(y)\Bigr\}.$
Therefore, the minmax problem can be written as 
\[
\min_{y\in[0,1]}\max_{\contsoldiers_1,x}U(\contsoldiers_1,x,y)
=\min_{y}\Bigl\{\max(y,\;2 - \contsoldiers_1^*(y))\Bigr\}.
\]
The unique minimizer \(y^*\in(0,1)\) solves $y = 2 - \contsoldiers_1^*(y)$. Replacing in the equality above for $\contsoldiers_1^*(y)$, we obtain the equation 
\[
y^4 -6y^3 +14y^2 -13y +4 =0,
\]
whose only root in \((0,1)\) is approximately $y^* \approx 0.647.$
At this value for \(y^*\), we obtain $\contsoldiers_1^* = \contsoldiers_1^*(y^*) \approx 1.355$ and $x^* = 0.$
Hence, the min-max value is $\min_{y}\max_{k,x}U(k,x,y)
= y^* \approx 0.647$.
\end{proof}

\subsection{Proof of~\Cref{thm:sion-one-sided-min-cont}}
\siononesidedmincont*
\begin{proof}
    We prove the minimax equality using the minimax theorem of Kneser and Fan presented in~\Cref{thm:kneser-fan}. The necessary definitions are stated under~\Cref{def:ccv-cvx-like}.

From the perspective of the minimizing player, the problem can be written as \[
\min_{\delta_1} \max_{\delta^2,\sigma^2} \min_{p\in\Delta(n)} \sum_{i\in [n]} \sum_{\alpha_i^1\in \mathcal A_i^1} \sum_{\alpha_i^2 \in \mathcal A_i^2} \delta_i^1(\alpha_i^2) \delta_{i,\contsoldiers_i^2}^2(\alpha_i^2)u_i(\alpha_i^1,\alpha_i^2,\contsoldiers_i^2),\]
where we convexify the objective function by rewriting the minimum over finitely many battlefields $i \in [n]$ as the minimum over convex combinations with respect to weight 
variables $p_i$. By deriving the dual of the inner minimization problem, we get \[\min_{\delta_1} \max_{\delta^2,\sigma^2} \max_{\lambda\in \mathbb R} \lambda \; \text{ s.t. } \;\lambda \leq \sum_{\alpha_i^1\in \mathcal A_i^1} \sum_{\alpha_i^2 \in \mathcal A_i^2} \delta_i^1(\alpha_i^1) \delta_{i,\contsoldiers_i^2}^2(\alpha_i^2)u_i(\alpha_i^1,\alpha_i^2,\contsoldiers_i^2)\; \forall i.
\]
Next, we define a compact space for $\lambda$. Since each $u_i(.)$ is continuous on the compact set $\Delta(\mathcal A_i^1) \times \Delta(\mathcal A_i^2) \times \Sigma^2,$ it attains a minimum and maximum.  Define
\[
  \underline u = \min_{i,\alpha_i^1,\alpha_i^2,\contsoldiers_i^2} u_i(\alpha_i^1,\alpha_i^2,\contsoldiers_i^2),
  \quad
  \overline u=\max_{i,\alpha_i^1,\alpha_i^2,\contsoldiers_i^2}u_i(\alpha_i^1,\alpha_i^2,\contsoldiers_i^2).
\]
Let $\lambda$ range over the closed interval $\Lambda = [\underline u,\;\overline u]$ which is compact in $\mathbb R$.  
Therefore, since each of $\Delta(\mathcal A_2), \Sigma^2,\Lambda$ is compact, their cartesian product, and hence the strategy space of the maximizing is compact. 

Next, letting $
T_i(\delta^1,\delta^2,\contsoldiers^2):=\sum_{\alpha_i^1\in \mathcal A_i^1} \sum_{\alpha_i^2 \in \mathcal A_i^2} \delta_i^1(\alpha_i^1)\delta_{i,\contsoldiers_i^2}^2(\alpha_i^2)u_i(\alpha_i^1,\alpha_i^2,\contsoldiers_i^2),$ we define \[
F\bigl(\delta^1,\delta^2,\contsoldiers^2,\lambda\bigr)
=
\begin{cases}
\;\lambda,
&\text{if }\lambda\le T_i(\delta^1,\delta^2,\contsoldiers^2)\ \forall i,\\
\;-\infty,
&\text{otherwise.}
\end{cases}
\]
Hence, our problem can be written as 
$ \min_{\delta^1}\max_{\delta^2,\contsoldiers^2,\lambda}
F(\delta^1,\delta^2,\contsoldiers^2,\lambda).$

We start by showing that $F$ is convexlike in $\delta^1$. Let $\delta^1_1,\delta^1_2\in\Delta(\mathcal A^1),\quad t\in(0,1)$,
and fix an arbitrary strategy of the maximizing player \((\delta^2,\contsoldiers^2,\lambda)\).  We must exhibit some
\(\tilde \delta^1\in\Delta(\mathcal A^1)\) such that
\[
F(\tilde \delta^1,\delta^2,\contsoldiers^2,\lambda)
\;\le\;
t\,F(\delta^1_1,\delta^2,\contsoldiers^2,\lambda)\;+\;(1-t)\,F(\delta^1_2,\delta^2,\contsoldiers^2,\lambda).
\]
There are two cases:
\begin{itemize}
  \item[\emph{(i)}] Both $\delta^1_1$ and $\delta^1_2$ are feasible, i.e.\ 
  $\lambda\le T_i(\delta^1,\delta^2,\contsoldiers^2)$ for all \(i\) and \(k=1,2\).  Let $\tilde \delta^1=t\delta^1_1+(1-t)\delta^1_2$ and note that $\tilde \delta^1 \in \Delta(\mathcal A^1)$.
  By linearity in $\delta^1$, we obtain
  \[
    T_i(\tilde \delta^1,\delta^2,\contsoldiers^2)
    = t\,T_i(\delta_1^1,\delta^2,\contsoldiers^2)+(1-t)\,T_i(x^2,y,\sigma_i)
    \ge \lambda,
  \]
  so $F(\tilde \delta^1,\delta^2,\contsoldiers^2,\lambda)=\lambda$ and $t\,F(\delta^1_1,\delta^2,\contsoldiers^2,\lambda)+(1-t)\,F(\delta^1_2,\delta^2,\contsoldiers^2,\lambda) = \lambda$, giving equality.
  
  \item[\emph{(ii)}] At least one of $\delta^1_1$ or $\delta^1_2$ is infeasible.  
  Without loss of generality, suppose $F(\delta^1_1,\delta^2,\contsoldiers^2,\lambda)=-\infty$.  Then
  \[ t\,F(\delta^1_1,\delta^2,\contsoldiers^2,\lambda)+(1-t)\,F(\delta^1_2,\delta^2,\contsoldiers^2,\lambda) = -\infty.
  \]
  Choose 
  $\tilde \delta^1 = \delta^1_1$.  Since $F(\tilde \delta^1,\delta^2,\contsoldiers^2,\lambda)=-\infty$, the inequality holds.
\end{itemize}
In both cases there exists $\tilde \delta^1\in\Delta(\mathcal A^1)$ with
$F(\tilde \delta^1,\delta^2,\contsoldiers^2,\lambda)\le t \,F(\delta^1_1,\delta^2,\contsoldiers^2,\lambda)+(1-t)\,F(\delta^1_2,\delta^2,\contsoldiers^2,\lambda),$
showing that $F$ is convexlike in $\delta^1$.

Next, we show that $F$ is concavelike in $(\delta^2,\contsoldiers^2,\lambda)$. Let 
$z_k=(\delta^2_k,\contsoldiers^2_k,\lambda_k), k=\{1,2\},$ and fix $t\in(0,1)$.  Define 
$\tilde \lambda = t\,\lambda_1 + (1-t)\,\lambda_2.$
We must exhibit some $\tilde z\in\Delta(\mathcal A^2) \times \Sigma^2 \times \Lambda$ such that for every $\delta^1$,
\[
t\,F(\delta^1,z_1)+(1-t)\,F(\delta^1,z_2)\;\le\;F(\delta^1,\tilde z).
\]
There are two cases:
\begin{itemize}
  \item[\emph{(i)}] Both $z_1$ and $z_2$ are feasible, i.e. $\lambda_k\le T_i(\delta^1,\delta^2_k,\contsoldiers^2_k)$ for all $i$ and $k=\{1,2\}$.  
  Then
  \[
    F(\delta^1,z_k)=\lambda_k,
    \quad
    t\,F(\delta^1,z_1)+(1-t)\,F(\delta^1,z_2)
    =\tilde \lambda.
  \]
  But $\tilde \lambda$ is a convex combination of $\lambda_1,\lambda_2$, so \[
  \tilde \lambda \leq \max\{\lambda_1,\lambda_2\} = \max\{F(\delta^1,z_1),F(\delta^1,z_2)\}
  \]
  Now choose
  \[
    \tilde z = 
    \begin{cases}
      z_1,&\text{if }\lambda_1\ge\lambda_2,\\
      z_2,&\text{otherwise.}
    \end{cases}
  \]
  By construction \(F(\delta^1,\tilde z)=\max\{\lambda_1,\lambda_2\}\ge\tilde \lambda\), so the inequality holds.

  \item[\emph{(ii)}] At least one of $z_1$ or $z_2$ is infeasible.  Then its $F$-value is $-\infty$, making
  \[
    t\,F(\delta^1,z_1)+(1-t)\,F(\delta^1,z_2) = -\infty \le F(\delta^1,z)
    \quad\forall\,z,
  \]
  so we can pick any $\tilde z$ and the inequality is trivial.
\end{itemize}
Thus in both cases there exists $\tilde z$ with $t\,F(\delta^1,z_1)+(1-t)\,F(\delta^1,z_2)\le F(\delta^1,\tilde z)$,
showing that $F$ is concavelike in $(\delta^2,\contsoldiers^2,\lambda)$.

Finally, it remains to show that $F$ is upper‑semicontinuous in $(\delta^2,\contsoldiers^2,\lambda)$.
Fix a strategy $\delta^1$ of the minimizing player.  For any $\beta \in \mathbb R$, the super‑level set
\[
\{(\delta^2,\contsoldiers^2,\lambda) : F(\delta^1,\delta^2,\contsoldiers^2,\lambda)\ge\beta\}
=
\{\lambda\ge\beta\}
\;\cap\;\bigcap_{i}\{\lambda\le T_i(\delta^1,\delta^2,\contsoldiers^2)\}
\]
is a finite intersection of closed sets, hence closed.  Thus $F$ is upper‑semicontinuous in $(\delta^2,\contsoldiers^2,\lambda)$. 

It follows that 
\[
\min_{\delta^1}\max_{\delta^2,\contsoldiers^2,\lambda}F(\delta^1,\delta^2,\contsoldiers^2,\lambda)
=
\max_{\delta^2,\contsoldiers^2,\lambda}\min_{\delta^1}F(\delta^1,\delta^2,\contsoldiers^2,\lambda).
\]
Unpacking both sides yields
\[
\min_{\delta^1}\max_{\delta^2,\contsoldiers^2} \min_i T_i =
\max_{\delta^2,\contsoldiers^2} \min_{\delta^1} \min_i T_i,
\]
which completes the proof.
\end{proof}

\subsubsection{Discussion: Inapplicability of the Kneser-Fan minimax theorem in other continuous settings}\label{discussion:KF}

In the sum-aggregated (both one-sided and two-sided) and the two-sided min-aggregated continuous two-level Blotto game, the Kenser-Fan theorem does not apply. It follows that the minimax equality is not satisfied and the Nash equilibrium does not exist. This validates the non‑existence results previously established for these cases (\Cref{thm:exist-two-sided-cont} and~\Cref{prop:exist-one-sided-sum-cont}).  

\paragraph{One-sided continuous two-level Blotto with sum aggregator.} We construct a counterexample showing that the payoff is not concavelike in the maximizing player’s strategy. 
    Consider a two-sided continuous Blotto game where the maximizing player (Player 2) has 1 (continuous) soldier to assign across battlefields, i.e. $\contsoldiers=1$ .  Let there be 2 battlefields $(2\times1)$ in which Player 2 plays trivial subgames (i.e. plays a single action w.p. 1). Let \[
    \begin{bmatrix}
    \contsoldiers_1 \\
    0
    \end{bmatrix}, \qquad 
    \begin{bmatrix}
    \contsoldiers_2 \\
    0
    \end{bmatrix}
    \]
    be the payoff matrices of Battlefields 1 and 2 respectively.
    Player 1, plays Top w.p. $x_i$ and Bottom w.p. $1-x_i$ in Battlefield $i\in\{1,2\}.$ Then, the total expected utility is given by \[
    f(x,\contsoldiers)=x_1\contsoldiers_1+x_2\contsoldiers_2
    \]
    The minimizer space is given by \[
    \Delta(\mathcal A^1) = \{(x_1,x_2)\in \mathbb R^2, x_1,x_2\ge 0\},
    \]
    and the maximizer's space of feasible assignments is given by 
    \[\ \Sigma^2=\{(\contsoldiers_1,\contsoldiers_2)\in\mathbb{R}^2 : \contsoldiers_1,\contsoldiers_2\ge0,\;\contsoldiers_1+\contsoldiers_2=1\bigr\}.\]
    By definition of a concavelike function, for any 2 points in the max player space, and any $t\in[0,1]$ there exists a strategy for the max player $\contsoldiers^2$ such that the inequality is satisfied \textit{for all $x \in \Delta(\mathcal A^1)$}. Let's choose \[
    \sigma' = (1,0),\quad \sigma''=(0,1),\quad t=\tfrac12.\]
    Then for every \(x\in[0,1]^2\),
    \[
    \frac12\,f\bigl(x,\sigma'\bigr)
    +\frac12\,f\bigl(x,\sigma''\bigr)
    =\frac12\bigl(x_1\cdot1 + x_2\cdot0\bigr)
    +\frac12\bigl(x_1\cdot0 + x_2\cdot1\bigr)
    =\tfrac12(x_1+x_2).
    \]
    $f$ being concavelike would require some $\bar\sigma\in\Sigma^2$ with
    \[
    f(x,\bar\sigma)
    =x_1\,\bar\sigma_1 + x_2\,\bar\sigma_2
    \;\ge\;\tfrac12(x_1+x_2)
    \quad\text{for all }x.
    \]
    Plugging $x=(1,0)$ forces $\bar\sigma_1\ge\tfrac12$, while
    $x=(0,1)$ forces $\bar\sigma_2\ge\tfrac12$.  But
    $\bar\sigma_1+\bar\sigma_2=1$ makes this impossible.  Hence $f$
    is not concavelike in $\contsoldiers^2$.

    It follows that the utility function fails to be concavelike in the two-sided sum-aggregated continuous Blotto, which justifies the non-existence of NE result in that setting.

 \paragraph{Two-sided continuous two-level Blotto with min aggregator.} 
     We outline the intuition behind why the Kneser-Fan minimax theorem fails to hold when both players allocate soldiers to battlefields under the min aggregator. 
     In this setting, the min-max formulation of the problem is given by \[
    \min_{\delta^1,\contsoldiers^1} \max_{\delta^2,\contsoldiers^2}\min_{i\in[n]} \sum_{\alpha_i^1,\alpha_i^2}\delta^1_i(\alpha_i^1)\delta_i^2(\alpha_i^2)u_i(\alpha_i^1,\alpha_i^2,\contsoldiers_i^1,\contsoldiers_i^2),
    \]
    which, after convexifying and taking the dual of the inner minimization problem, can be written as 
    \[\min_{\delta^1,\contsoldiers^1}\max_{\delta^2,\contsoldiers^2,\lambda}
    F(\delta^1,\contsoldiers^1,\delta^2,\contsoldiers^2,\lambda),\]
    where \[
    F\bigl(\delta^1, \contsoldiers^1, \delta^2,\contsoldiers^2,\lambda\bigr) =
    \begin{cases}
    \;\lambda,
    &\text{if } \;\lambda\le T_i(\delta^1,\contsoldiers^1,\delta^2,\contsoldiers^2) \; \forall i,\\
    \;-\infty,
    &\text{otherwise.}
    \end{cases}
    \]
    with $T_i(\delta^1,\delta^2,\contsoldiers^1,\contsoldiers^2) := \sum_{\alpha_i^1,\alpha_i^2}\delta^1_i(\alpha_i^1)\delta_i^2(\alpha_i^2)u_i(\alpha_i^1,\alpha_i^2,\contsoldiers_i^1,\contsoldiers_i^2).$
    First, note that convexlikeness in the $(\delta^1,\contsoldiers^1)$–block demands that for \textit{any} two minimizer choices $(\delta^1_1,\contsoldiers^1_1)$ and $(\delta^1_2,\contsoldiers^1_2)$ and \textit{any} weight $t\in[0,1]$, we can find some $(\tilde \delta^1,\tilde \contsoldiers^1)$ so that \[
    tF(\delta^1_1,\contsoldiers^1_1,\delta^2,\contsoldiers^2,\lambda)+(1-t)F\bigl(\delta^1_2, \contsoldiers^1_2, \delta^2, \contsoldiers^2,\lambda\bigr) \geq F(\tilde \delta^1, \tilde \contsoldiers^1, \delta^2,\contsoldiers^2,\lambda) \; \forall (\delta^2,\contsoldiers^2,\lambda)
    \]
    We saw that the linearity in $\delta^1$ lets us pick $\tilde \delta^1=t\delta^1_1+(1-t)\delta^1_2$, and that takes care of the $\delta^1$ coordinate. But nothing in our assumptions guarantees that we can choose a $\tilde \contsoldiers^1$ that preserves $\lambda\le T_i(\tilde \delta^1,\tilde \contsoldiers^1,\delta^2,\contsoldiers^2)$ whenever it held at $\contsoldiers^1_1$ and $\contsoldiers^1_2$. More precisely, for a fixed $(\tilde \delta^1, \delta^2, \contsoldiers^2,\lambda)$, the set of feasible $\contsoldiers^1$ is given by \[
    S=\{\contsoldiers^1: T_i(\tilde \delta^1, \contsoldiers^1,\delta^2,\contsoldiers^2)\ge \lambda \; \forall i\}. 
    \]
    However, if $T_i$ is not convex in $\contsoldiers^1$ then $S$ need not be a convex (or even connected) region. 
    Moreover, convexlikeness does not just demand some $\tilde \contsoldiers^1\in S$ for one fixed choice of $(\delta^2,\contsoldiers^2,\lambda)$, it requires a single $\tilde \contsoldiers^1$ that works simultaneously \textit{for every} triple $(\delta^2,\contsoldiers^2,\lambda)$. In other words, $\tilde \contsoldiers^1$ must lie in the intersection \[
    \bigcap_{(\delta^2,\contsoldiers^2,\lambda)} \{\contsoldiers^1: T_i(\tilde \delta^1, \contsoldiers^1,\delta^2,\contsoldiers^2)\ge \lambda \; \forall i\}.
    \]
    Since that intersection can be empty, no such $\tilde \contsoldiers^1$ is guaranteed to exist, and one can construct examples of payoff functions where convexlikeness fails.

\subsection{Proof of~\Cref{thm:ne-comp-one-sided-cont-min}}
\componesidedmincontlinear*
\begin{proof}
    We can write the min-max formulation of the problem as:
\begin{align}
    &\max_{\delta^{2},\sigma^2}\min_{\delta^{1}} \min_{i\in[n]} \mathbb{E}[u_i(\alpha_i^{1},\alpha_i^{2},\contsoldiers_i^2)] \\ &= \max_{\contsoldiers^2}\max_{\delta^{2}}\min_{\delta^{1}}\min_{p\in\Delta(n)} \sum_{i\in[n]} p_i\biggl(\sum_{\alpha_i^1\in \mathcal A_i^1} \sum_{\alpha_i^2\in \mathcal A_i^2} \delta_{i}^1(\alpha_i^1)\delta^2_{i,\contsoldiers_i^2}(\alpha_i^2) u_i(\alpha_i^1,\alpha_i^2,\contsoldiers_i^2) \biggr) \label{eq2:cont-min-one-sided} \\
    &= \max_{\contsoldiers^2}\max_{\delta^{2}}\min_{y^1}\sum_{i\in[n]} \sum_{\alpha_i^1\in \mathcal A_i^1} \sum_{\alpha_i^2\in \mathcal A_i^2} y^1(i,\alpha_i^1)\delta^2_{i,\contsoldiers_i^2}(\alpha_i^2) c_i\cdot\contsoldiers_i^2 \label{eq3:cont-min-one-sided} \\
    &= \max_{y^{2}}\min_{y^1}\sum_{i\in[n]} \sum_{\alpha_i^1\in \mathcal A_i^1} \sum_{\alpha_i^2\in \mathcal A_i^2} y^1(i,\alpha_i^1) y^2_{\contsoldiers_i^2}(i,\alpha_i^2) c_i\label{eq4:cont-min-one-sided} 
\end{align}
where equality \eqref{eq2:cont-min-one-sided} follows from convexifying the objective function by rewriting the minimum over finitely many battlefields $i\in[n]$ as the minimum over convex combinations with respect to weight variables $p_i$. Note that this convexification step arise from the fact that the optimum is always attained at a vertex of the simplex.  Equation \eqref{eq3:cont-min-one-sided} then comes from merging the variables $p,\delta^{1}$ into a sequence form product variable $y^1$ defined over the polytope $\mathcal P$. 
Similarly, in equation \eqref{eq4:cont-min-one-sided}, we merge the variables $\contsoldiers^2,\delta^2$ into the sequence-form variable $y^2_{\contsoldiers^2}$ is defined over polytope $\mathcal Q$ (refer to Sections ~\ref{sect:disc-min} and~\ref{sect:cont-sum} for more details on polytopes $\mathcal P$ and $\mathcal Q$ and the corresponding strategy representations). 
\end{proof}

\subsection{Proof of~\Cref{cor:one-sided-min-linear-cont} }
\existonesidedminlinear*

\begin{proof}
As the objective is bilinear in the sequence form variables $y^1$ and $y^2$, a Nash equilibrium of the game can be obtained by solving a linear program (see~\Cref{dual:one-sided-cont-min}). Since we are optimizing over the polytopes $\mathcal P$ and $\mathcal Q$, both of which have polynomial-size representations, then the solution can be computed in polynomial time.  
\end{proof}

\subsection{Proof of~\Cref{thm:maxmin-one-sided-min-cont-linear}}
\maxminonesidedmincont*
\begin{proof}
We can formulate the max-min problem as follows:
\begin{align}
    &\max_{\delta^{2},\sigma^2}\min_{\delta^{1}} \min_{i\in[n]} \mathbb{E}[u_i(\alpha_i^{1},\alpha_i^{2},\contsoldiers_i^2)] \\ &= \max_{\contsoldiers^2}\max_{\delta^{2}}\min_{\delta^{1}}\min_{p\in\Delta(n)} \sum_{i\in[n]} p_i\biggl(\sum_{\alpha_i^1\in \mathcal A_i^1} \sum_{\alpha_i^2\in \mathcal A_i^2} \delta_{i}^1(\alpha_i^1)\delta^2_{i,\contsoldiers_i^2}(\alpha_i^2) u_i(\alpha_i^1,\alpha_i^2,\contsoldiers_i^2) \biggr) \label{eq2:cont-min-one-sided-maxmin} \\
    &= \max_{\contsoldiers^2}\max_{\delta^{2}}\min_{y^1}\sum_{i\in[n]} \sum_{\alpha_i^1\in \mathcal A_i^1} \sum_{\alpha_i^2\in \mathcal A_i^2} y^1(i,\alpha_i^1)\delta^2_{i,\contsoldiers_i^2}(\alpha_i^2) u_i(\alpha_i^1,\alpha_i^2,\contsoldiers_i^2) \label{eq3:cont-min-one-sided-maxmin} \\
    &= \max_{\contsoldiers^2}\min_{y^1} \max_{\delta^{2}} \sum_{i\in[n]} \sum_{\alpha_i^1\in \mathcal A_i^1} \sum_{\alpha_i^2\in \mathcal A_i^2} y^1(i,\alpha_i^1)\delta^2_{i,\contsoldiers_i^2}(\alpha_i^2) u_i(\alpha_i^1,\alpha_i^2,\contsoldiers_i^2) \label{eq4:cont-min-one-sided-maxmin} \\ 
    &= \max_{\sigma^2}\min_{p\in\Delta(n)}\sum_{i\in[n]} p_i \min_{\delta_i^1}\max_{\delta_i^2} \sum_{\alpha_i^1\in \mathcal A_i^1} \sum_{\alpha_i^2\in \mathcal A_i^2} \delta^1_i(\alpha_i^1) \delta^2_{i,\contsoldiers_i^2}(\alpha_i^2) u_i(\alpha_i^1,\alpha_i^2,\contsoldiers_i^2)\label{eq5:cont-min-one-sided-maxmin} \\
    &= \max_{\sigma^2}\min_{p\in\Delta(n)}\sum_{i\in[n]} p_i v_i^*(\sigma^2_i)\\ 
    &= \max_{\sigma^2}\min_{i\in[n]}v^*_i(\contsoldiers_i^2), 
\end{align}
where equations \eqref{eq2:cont-min-one-sided-maxmin} and \eqref{eq3:cont-min-one-sided-maxmin} follow from convexifying the objective function and merging variables $p,\delta^1$ into $y^1$ defined over the polytope $\mathcal P$. Since the objective is now bilinear in $y^1$ and $\delta^2$, we now apply Sion's minimax theorem in equation \eqref{eq4:cont-min-one-sided-maxmin} to exchange the inner $\max$ and $\min$. In equation \eqref{eq5:cont-min-one-sided-maxmin}, we disaggegate $y^1$ back into $p,\delta^1$  and rearrange the formula. Finally, having a bilinear objective in $\delta^1,\delta^2$, we apply Sion's minimax theorem to the inner minmax problem and replace  $\min_{\delta^1}\max_{\delta^2}\mathbb E[u_i(\alpha_i^1,\alpha_i^2,\sigma_i^2)]$ by the Nash value function $v_i^*(\sigma^2_i)$.

We now clarify why the optimization over $\delta^1$ and $\delta^2$ in Equation~\eqref{eq5:cont-min-one-sided-maxmin} can be disaggregated across battlefields. In fact, once we move the $\min$ and $\max$ operators inside the sum over 
$i\in[n]$, the expression becomes separable: the term for each battlefield $i$ involves only local strategies $\delta^1_i$ and $\delta_i^2$. This separation is valid because battlefield subgames are assumed to be independent, so players choose strategies independently in each subgame. As a result, we can optimize each subgame separately and compute the Nash value $v_i^*(\sigma^2_i)$ for each battlefield $i$ in isolation.  
\end{proof}

\subsection{Proof of~\Cref{lemma:quasiconcavity}}
\quasiconcavitymaxmin*
\begin{proof}
    Since $u_i(\alpha_i^1,\alpha_i^2,\sigma^2)$ is increasing in $\sigma^2$ for every $\alpha_i^1$ and $\alpha_i^2$, the Nash equilibrium computation preserves this monotonicity, implying that $v^*_i(\sigma^2)$ is increasing in $\sigma^2$. Since every monotonic univariate function is quasiconcave, and the pointwise minimum of quasiconcave functions is quasiconcave, then $V(\sigma^2)$ is quasiconcave. 
\end{proof}

\subsection{Proof of~\Cref{prop:nash-subgrad}}
\subgnashval*
\begin{proof}
  Recall the definition of the Nash value function
  \[
    v_i^*(\sigma^2)
    = \max_{\delta^2 \in \Delta(\mathcal A_i^2)}
      \min_{\delta^1 \in \Delta(\mathcal A_i^1)}
      \sum_{\alpha_i^1,\alpha_i^2}
        \delta_i^1(\alpha_i^1)\,\delta_i^2(\alpha_i^2)\,
        u_i(\alpha_i^1,\alpha_i^2,\sigma^2_i).
  \]

  \textit{1. Inner minimization over \(\delta^1\).}  
  
  For each fixed \(\delta_i^2\), the map
  \(\delta_i^1\mapsto \sum_{\alpha_i^1,\alpha_i^2}\delta_i^1(\alpha_i^1)\delta_i^2(\alpha_i^2)u_i(\alpha_i^1,\alpha_i^2,\sigma^2_i)\)
  is linear in \(\delta_i^1\) and hence its minimum over the simplex
  \(\Delta(\mathcal A_i^1)\) is attained at an extreme point (pure action)
  \(\alpha_i^1\). 
  Define 
  \[
    \phi_{\alpha_i^2}(\sigma_i^2)
    := \min_{\alpha_i^1\in\mathcal A_i^1}
       u_i(\alpha_i^1,\alpha_i^2,\sigma^2_i),
    \quad
    \forall\,\alpha_i^2\in\mathcal A_i^2.
  \]
  Then we have
  \[
    \min_{\delta^1}\sum_{\alpha_i^1,\alpha_i^2}
      \delta_i^1(\alpha_i^1)\delta_i^2(\alpha_i^2)u_i(\alpha_i^1,\alpha_i^2,\sigma^2_i)
    = \sum_{\alpha_i^2\in \mathcal A_i^2} \delta_i^2(\alpha_i^2) \min_{\alpha^1\in\mathcal A_i^1} u_i(\alpha_i^1,\alpha_i^2,\sigma^2_i)
    = \sum_{\alpha_i^2\in \mathcal A_i^2} \delta_i^2(\alpha_i^2) \phi_{\alpha_i^2}(\sigma_i^2).
  \]
Since each map \(u_i(\alpha_i^1,\alpha_i^2,\cdot)\) is Lipschitz,
  \(\phi_{\alpha_i^2}\) is a pointwise minimum of finitely many Lipschitz
  functions and therefore subdifferentiable.  It follows that, for each \(\alpha_i^2\), we have
  \[
    \partial_{\sigma_i^2}\phi_{\alpha_i^2}(\sigma^2)
    = \conv\Bigl\{
        \partial_{\sigma_i^2}u_i(\alpha_i^1,\alpha_i^2,\sigma^2_i)
        : \alpha_i^1\in\arg\min_{\beta_i^1\in\mathcal A_i^1}
          u_i(\beta_i^1,\alpha_i^2,\sigma^2_i)
      \Bigr\}.
  \]

\textit{2. Outer maximization over \(\delta^2\).}
  Define
  \[
    F(\sigma^2,\delta^2)
    := \sum_{\alpha_i^2\in\mathcal A_i^2}
         \delta_i^2(\alpha_i^2)\,
         \phi_{\alpha_i^2}(\sigma^2).
  \]
  Since \(F\) is linear in \(\delta^2\), its maximum over the simplex
  \(\Delta(\mathcal A_i^2)\) is attained at a vertex (pure action)
  \(\alpha_i^2\).  Hence
  \[
    v_i^*(\sigma^2)
    = \max_{\delta^2} F(\sigma^2,\delta^2)
    = \max_{\alpha_i^2\in\mathcal A_i^2}
      \phi_{\alpha_i^2}(\sigma^2).
  \]
Each \(\phi_{\alpha_i^2}\) is Lipschitz (as a pointwise minimum of Lipschitz maps), so by Danskin’s theorem,
  \[
    \partial_{\sigma_i^2}v_i^*(\sigma^2)
    = \conv\Bigl\{
        \partial_{\sigma_i^2}\phi_{\alpha_i^2}(\sigma^2)
        : \alpha_i^2\in\arg\max_{\beta_i^2\in\mathcal A_i^2}
          \phi_{\beta_i^2}(\sigma_i^2)
      \Bigr\}.
  \]
\textit{3. Substitute the inner‐minimum subgradients.}
Substituting the expression for  $\partial_{\sigma_i^2}\phi_{\alpha_i^2}(\sigma^2)$ into $\partial_{\sigma_i^2}v_i^*(\sigma^2)$ yields $\partial_{\sigma_i^2}v_i^*(\sigma^2)=$
  \begin{equation}\label{eq:subdiff-u}
     \conv\Bigl\{
        \partial_{\sigma_i^2}u_i(\alpha_i^1,\alpha_i^2,\sigma^2_i):
        \alpha_i^2\in\arg\max_{\beta_i^2}\min_{\beta_i^1}u_i(\beta_i^1,\beta_i^2,\sigma^2_i),
        \; \alpha_i^1\in\arg\min_{\beta_i^1}u_i(\beta_i^1,\alpha_i^2,\sigma^2_i)
      \Bigr\}.
\end{equation}

Let \[
S(\sigma^2):= \biggl\{(\alpha_i^1,\alpha_i^2): \alpha_i^2\in\arg\max_{\beta_i^2}\min_{\beta_i^1}u_i(\beta_i^1,\beta_i^2,\sigma^2_i),
        \; \alpha_i^1\in\arg\min_{\beta_i^1}u_i(\beta_i^1,\alpha_i^2,\sigma^2_i)\biggr\}
\]
denote the set of pure pairs $\alpha_i^2\in\arg\max_{\beta_i^2}\min_{\beta_i^1}u_i(\beta_i^1,\beta_i^2,\sigma^2_i)$,  
$\alpha_i^1\in\arg\min_{\beta_i^1}u_i(\beta_i^1,\alpha_i^2,\sigma^2_i)$. 
  By~\eqref{eq:subdiff-u}, we have 
  \[
    \partial_{\sigma^2_i}v_i^*(\sigma^2)
    = \conv\Bigl\{
      \partial_{\sigma^2_i}u_i(\alpha_i^1,\alpha_i^2,\sigma^2_i):(\alpha_i^1,\alpha_i^2)\in S(\sigma^2)
    \Bigr\}.
  \]
  
By the equilibrium properties of the zero-sum subgame $G_i$, we know that every pair \((\alpha_i^1,\alpha_i^2)\) in the support of \((\delta_i^{1*},\delta_i^{2*})\) lies in \(S(\sigma^2)\).  Since the weights 
  \(\delta_i^{1*}(\alpha_i^1),\delta_i^{2*}(\alpha_i^2)\)  
  form a convex combination supported on \(S(\sigma^2)\), the vector
  \(\sum_{\alpha_i^1,\alpha_i^2} \delta_i^{1*}(\alpha_i^1)\,\delta_i^{2*}(\alpha_i^2)\,z_{\alpha_i^1,\alpha_i^2}\)  
  is exactly one element of that convex hull.  It follows that \(g_i\in\partial_{\sigma^2_i}v_i^*(\sigma^2)\).  
\end{proof}

\subsection{Proof of~\Cref{cor:V-subgrad}}
\subgaggfct*
 \begin{proof}

Under the assumption that each $u_i$ is Lipschitz, for each battlefield $i$ and pair of actions $(\alpha_i^1,\alpha_i^2)$ there exists a constant $L_i^{\alpha_i^1,\alpha_i^2}$ such that
    \[\bigl|u_i(\alpha_i^1,\alpha_i^2,\sigma^2_i)-u_i(\alpha_i^1,\alpha_i^2,\tilde\sigma^2_i)\bigr|\leq L_i^{\alpha_i^1,\alpha_i^2}\,\bigl\|\sigma^2-\tilde\sigma^2\bigr\|
    \quad\forall\,\sigma^2,\tilde\sigma^2.
    \]
    Fix Nash equilibrium strategies \(\delta_i^{1*},\delta_i^{2*}\).  Then
    \[
      v_i^*(\sigma^2)
      = \sum_{\alpha_i^1,\alpha_i^2}
          \delta_i^{1*}(\alpha_i^1)\,\delta_i^{2*}(\alpha_i^2)\;u_i\bigl(\alpha_i^1,\alpha_i^2,\sigma^2_i\bigr).
    \]
    For any \(\sigma^2,\tilde\sigma^2\),
    \begin{align*}
      \bigl|v_i^*(\sigma^2)-v_i^*(\tilde\sigma^2)\bigr|
      &\le \sum_{\alpha_i^1,\alpha_i^2}
            \delta_i^{1*}(\alpha_i^1)\,\delta_i^{2*}(\alpha_i^2)\,
            \bigl|u_i(\alpha_i^1,\alpha_i^2,\sigma^2_i)-u_i(\alpha_i^1,\alpha_i^2,\tilde\sigma^2_i)\bigr|\\
      &\le \sum_{\alpha_i^1,\alpha_i^2}
            \delta_i^{1*}(\alpha_i^1)\,\delta_i^{2*}(\alpha_i^2)\,
            L_i^{\alpha_i^1,\alpha_i^2}\,
            \|\sigma^2-\tilde\sigma^2\|\\
      &\le \Bigl(\max_{\alpha_i^1,\alpha_i^2}L_i^{\alpha_i^1,\alpha_i^2}\Bigr)\,
            \|\sigma^2-\tilde\sigma^2\|
      =:L_i\,\|\sigma^2-\tilde\sigma^2\|.
    \end{align*}
    where we set $L_i =\max_{\alpha_i^1,\alpha_i^2}L_i^{\alpha_i^1,\alpha_i^2}$. It follows that each $v_i^*$ is $ L_i$-Lipschitz.
    
    Next, since 
    \(\;V(\sigma^2)=\min_i v_i^*(\sigma^2)\), 
    and each \(v_i^*\) is Lipschitz with constant \(L_i\), the pointwise minimum of these \(n\) functions is Lipschitz with constant $L =\max_{i\in [n]}L_i.$
    Hence, we can write
    \[
      \bigl|V(\sigma^2)-V(\tilde\sigma^2)\bigr| \leq
      L\,\|\sigma^2-\tilde\sigma^2\|.
    \]
    Hence, $V$ is $L$-Lipschitz on the compact set $\Sigma^2$. Moreover, Weistrass' theorem shows that an optimal solution exists: there exists \(\sigma^{2*}\in\Sigma^2\) with
    \(\max_{\sigma^2}V(\sigma^2)=\max_{\sigma^2}\min_{i\in[n]}v^*_i(\sigma^2) = V(\sigma^{2*})\). 

  Let
  \[
    W(\sigma^2)\;=\;-\,V(\sigma^2)
    =\max_{i\in[n]}\bigl[-\,v_i^*(\sigma^2)\bigr].
  \]
  Since each \(v_i^*\) is Lipschitz on \(\Sigma^2\), hence so is 
  \(w_i(\sigma^2):=-v_i^*(\sigma^2)\).  Since \(W\) is the pointwise maximum of the finitely many Lipschitz functions \(\{w_i\}\), Danskin’s theorem yields
  \[
    \partial_{\sigma^2}W(\sigma^2)
    = \conv\Bigl\{
        \partial_{\sigma^2}w_i(\sigma^2)
        : i\in\arg\max_{j}\,w_j(\sigma^2)
      \Bigr\}.
  \]
  Noting that 
  \(\arg\max_j w_j = \arg\min_j v_j^*\) and 
  \(\partial_{\sigma^2}w_i = -\,\partial_{\sigma^2}v_i^*\), we obtain
  \[
    \partial_{\sigma^2}V(\sigma^2)
    = -\,\partial_{\sigma^2}W(\sigma^2)
    = \conv\Bigl\{
        \partial_{\sigma^2}v_i^*(\sigma^2)
        : i\in\arg\min_{j}v_j^*(\sigma^2)
      \Bigr\},
  \]
  which completes the proof.
\end{proof}

\subsection{Proof of~\Cref{prop:sharp-max}}
\sharpmax*
\begin{proof}
    Assume that battlefield utilities are affine functions in $\contsoldiers^2$. 
    For any fixed $\contsoldiers^2$, the expected payoff to Player 2 under mixed strategies $(\delta_i^1,\delta_i^2)$ in subgame $G_i$ is \[
  \mathbb E [u_i(\alpha_i^1,\alpha_i^2,\contsoldiers_i^2)] = \sum_{\alpha_i^1,\alpha_i^2} \delta_i^1(\alpha_i^1)\delta_i^2(\alpha_i^2) c_{i_{\alpha_i^1,\alpha_i^2}}\;\contsoldiers^2_i+d_{i_{\alpha_i^1,\alpha_i^2}}
  \] 
  Hence \[v_i^*(\contsoldiers^2_i) = \biggl[ \max_{\delta^2} \min_{\delta^1} \sum_{\alpha_i^1,\alpha_i^2} \delta_i^1(\alpha_i^1)\delta_i^2(\alpha_i^2) c_{i_{\alpha_i^1,\alpha_i^2}} \biggr]\contsoldiers^2_i + \underbrace{ \sum_{\alpha_i^1,\alpha_i^2} \delta_i^1(\alpha_i^1)\delta_i^2(\alpha_i^2) d_{i_{\alpha_i^1,\alpha_i^2}}}_{\text{constant in }\contsoldiers_i^2} =C_i \;\contsoldiers^2_i + D_i.\]
  It follows that $v_i^*(\contsoldiers^2_i)$ is affine with slope $C_i$. Since each $c_{i_{\alpha_i^1,\alpha_i^2}}>0$, every such slope is at least $\epsilon=\min_{i,\alpha_i^1,\alpha_i^2}c_{i_{\alpha_i^1,\alpha_i^2}}>0$. Moreover, since the pointwise minimum of a finite family of affine functions is concave, then $V=\min_i v_i^*$ is concave on the soldiers simplex $\Sigma^2$. The usual subgradient inequality at $\contsoldiers^{2*}$ gives, for any $g\in\partial V(\contsoldiers^{2*})$ and any $\contsoldiers^2\in \Sigma^2$, \[
   V(\contsoldiers^{2*})-V(\contsoldiers^2) \ge \langle g, \; \contsoldiers^{2*} - \contsoldiers^2 \rangle.
  \]
  Because $v_i^*=C_i\; \contsoldiers^{2*}_i+D_i$ with $C_i\ge \epsilon > 0$, the subdifferential $\partial_{\contsoldiers^2}v_i^*$ at $\contsoldiers^{2*}$ is simply the singleton $\{C_i\}$. Hence, using~\Cref{cor:V-subgrad}, it follows that one valid choice of $g$ is the vector $(C_{i^*}\cdot e_i)^T, \; C_{i^*}\geq \epsilon$. Plugging this into the subgradient inequality, we obtain \[V(\contsoldiers^{2*})-V(\contsoldiers^2) \ge C_{i^*}(\contsoldiers^{2*}_{i^*}-\contsoldiers^2_{i^*})=C_{i^*}\;\|\contsoldiers^{2*}-\contsoldiers^2\|_{\infty} \ge \epsilon \;\|\contsoldiers^{2*}-\contsoldiers^2\|_{\infty}\] where the equality follows since the optimum is attained at a vertex of the simplex $\Sigma^2$. This establishes the desired sharpness inequality with $p=1$ and $\|.\|=\|.\|_{\infty}$.
\end{proof}

\section{Experiments}
\label{sec:appendix-expts}
\subsection{Discrete two-sided with sum aggregator}
\paragraph{Experimental Setup} The experiments were run on a Apple M1 processor with 16GB of RAM. We used Gurobi \cite{gurobi} as the LP solver. The online learners were implemented in C++ and built with the \textit{-Ofast} flag to optimize for running time. We did not employ OpenMP or any other parallelization. Note also that we did not resort to complicated optimizations (particularly with respect to reuse of memory) in our implementation.  

\paragraph{Game Description} We have $n$ battlefields. We adopt a ``soft'' version of blotto where the probability that player 1 wins battlefield $i$ is given by $k^1_i/(k^1_i + k^2_i)$, and the battlefield itself is worth $1/n$ (therefore battlefields are heterogeneous). In other words, for each battlefield a soldier (from either side) is chosen uniformly at random to be the winner, and the player it belongs to is awarded victory of the battlefield. If both sides field no soldiers, then a random player wins it. Note that losers get negative the worth of the battlefield, so this is a zero-sum game. 

Within each battlefield, each player has (regardless of how many soldiers it places in it) two actions, to double or not. Doubling means that winning yields double the utility, but also double the loss if it loses the battlefield. If both players double, then the value of the battlefield is quadrupled. Therefore, intuitively players should double when they have placed more soldiers in the battlefield. Note that this doubling battlefield subgame was chosen for simplicity; one can easily think of more complicated settings.

\paragraph{Extensions to the main paper.} In Table~\ref{tab:expt-rm-extended} we present extended results to to the main paper. We run additional experiments using regular regret matching(+), as well as their optimistic variants predictive regret matching (PRM) and PRM+. PRM and PRM+ as regret minimizers over the simplex (see Section~\ref{sec:online-learning-appendix-discrete-sum-agg}). Here, we use vanilla variants where players are performing updates simultaneously (as opposed to alternating variants), and furthermore, current strategies are obtained using simple averaging over past iterates. We terminate the first order methods when the saddle point gap is below $2.0 \cdot 10^{-3}$. This computation of saddle-point gap and check is done every $100$ iterations. 
\begin{table}[h]
    \centering
    \begin{tabular}{cccc|ccccc}
        \hline
        $n$ & $m^1$ & $m^2$ & $|\Gamma^1|, |\Gamma^2|$ & LP (s) & RM+ (s)& RM (s) & PRM+ (s) & PRM (s) \\
        \hline
        30 & 100 & 50 & 85, 49 & 91 & 56 &64  & 50 & 49 
        \\
        \hline
        35 & 125 & 70 & 281k, 92k & $4.8\cdot 10^3$ & $1.2\cdot 10^2$ & $1.5\cdot 10^2$ & $1.1\cdot 10^2$ & $1.0\cdot 10^2$ 
        \\
        \hline 
        40 & 150 & 100 & 1M, 262k& $5.4 \cdot 10^3$ & $6.5\cdot 10^2$& $6.5\cdot 10^2$& $4.5\cdot 10^2$ & $4.4\cdot 10^2$
        \\
        \hline 
        50 & 200 & 100 & 12.6M, 2.05M & NA  & $9.5\cdot 10^3$ & $1.2 \cdot 10^3$& $8.2 \cdot 10^3$ & $7.2 \cdot 10^3$
        \\
        \hline
    \end{tabular}
    \caption{Runtime for discrete soldiers with sum aggregators}
    \label{tab:expt-rm-extended}
\end{table}

We observe that the predictive variants perform better than. We also observed that even though running times for PRM+ are slightly slower than PRM, the number of iterates required is slightly lower as well, suggesting that the slower runtime is probably due to the extra computation required per iteration.

\paragraph{Additional results.} In Table~\ref{tab:expt-rm-speedup} we show the results from two speedups. First, we adopt alternating updates, which have been shown empirically to perform significantly better than simultaneous updates. Second, we adopt \textit{quadratic averaging} over past iterates for current strategies. This has been shown to also converge to an equilibrium, but often at a much faster rate since more recent iterates are given a higher weightage. Note that the running times for LPs are reused and exactly the same as in Table~\ref{tab:expt-rm-extended}.
\begin{table}[h]
    \centering
    \begin{tabular}{cccc|ccccc}
        \hline
        $n$ & $m^1$ & $m^2$ & $|\Gamma^1|, |\Gamma^2|$ & LP (s) & RM+ (s)& RM (s) & PRM+ (s) & PRM (s) \\
        \hline
        30 & 100 & 50 & 85, 49 & 91 & 5.1 &6.1 & 4.2 & 4.0 
        \\
        \hline
        35 & 125 & 70 & 281k, 92k & $4.8\cdot 10^3$ & 12 & 18 & 11 & 16 
        \\
        \hline 
        40 & 150 & 100 & 1M, 262k& $5.4 \cdot 10^3$ & 58 & 90& 47& 52
        \\
        \hline 
        50 & 200 & 100 & 12.6M, 2.05M & NA  & $1.2\cdot 10^3$ & $1.3\cdot 10^3$ & $9.3 \cdot 10^2$ & $1.2\cdot 10^3$
        \\
        \hline
    \end{tabular}
    \caption{Runtime for discrete soldiers with sum aggregators and speedups}
    \label{tab:expt-rm-speedup}
\end{table}

\subsection{Continuous one-sided with min-aggregator}\label{app:cont-one-sided-experiments}

\paragraph{Experimental Setup} The experiments were run on a 12th Gen Intel(R) Core(TM) i9-12900 processor with 67GB of RAM. We used Gurobi~\cite{gurobi} as the LP solver to compute the subgame Nash equilibrium in every battlefield.

\subsubsection{Random battlefield utilities}

Consider a min-aggregated continuous two-level Blotto game, under the one-sided scenario in which Player 2 is endowed with 20 soldiers that she allocates across 5 battlefields. We randomly generate 2 families of utility functions in every battlefield: (i) affine utilities where $u_i(\contsoldiers^2)=c_i\,\contsoldiers^2+d_i$, and (ii) quadratic utilities where $u_i(\contsoldiers^2)=b_i\,(\contsoldiers^2)^2+c_i\,\contsoldiers^2+d_i$, with $b_i,c_i,d_i\sim\mathrm{Uniform}[0,100]$. We consider different subgame sizes, varying players' action spaces in $\{20,30,40\}$ and generate 10 independent instances per size. We run~\Cref{alg:psa} using a diminishing step-size $\eta_t = \eta_0 / \sqrt{(t + 1)}$ with $\eta_0=0.01$.
Figures \ref{fig:affine-indiv-values} and \ref{fig:quadr-indiv-values} display, for each subgame size, the normalized aggregate objective $V^t$ (red curve) together with the individual trajectories from the 10 independent instances (gray curves), all plotted over running time (in seconds). Here $V_{\mathrm{norm}}^t = \frac{V^t}{V_{\max}},$ where $V_{\max}$ is highest value attained across instances.

\begin{figure}[h]
    \centering
   \includegraphics[width=\linewidth]{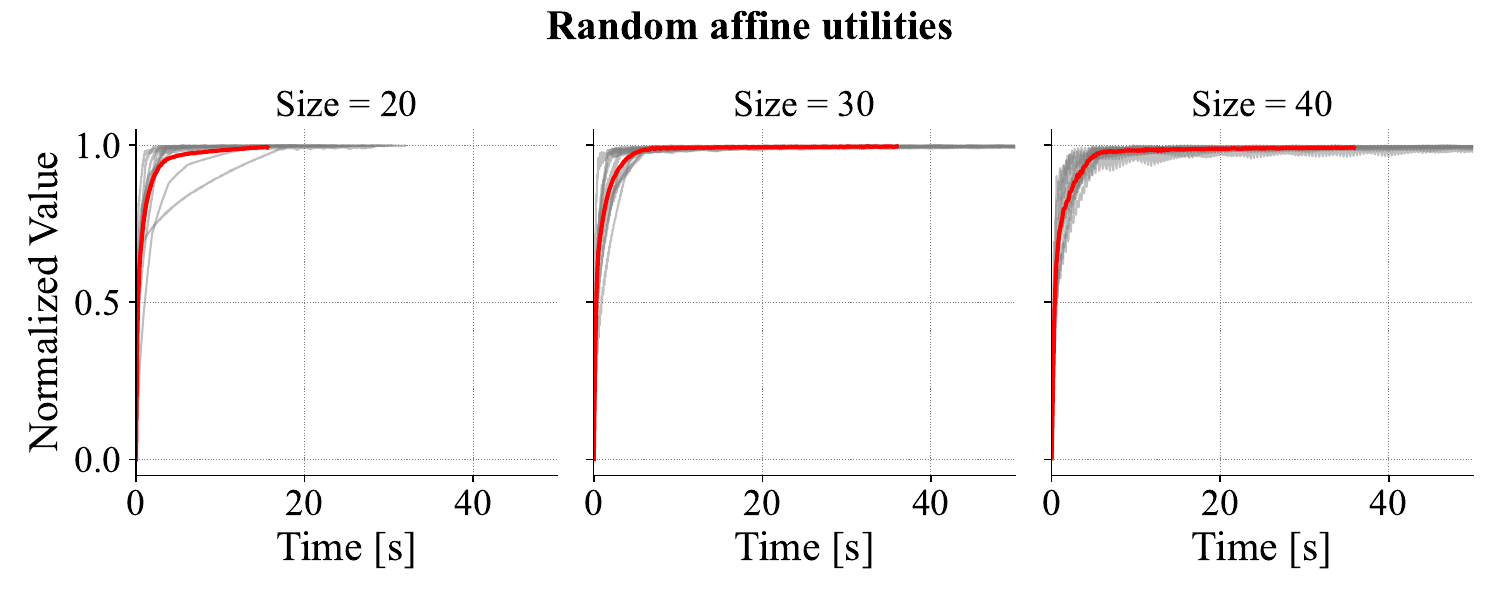}
\caption{Average (red) normalized objective value over 10 instances (gray) of a one-sided continuous two-level Blotto game with min aggregator and randomly generated affine battlefield utilities.}
\label{fig:affine-indiv-values}
\end{figure}

\begin{figure}[h]
    \centering
   \includegraphics[width=\linewidth]{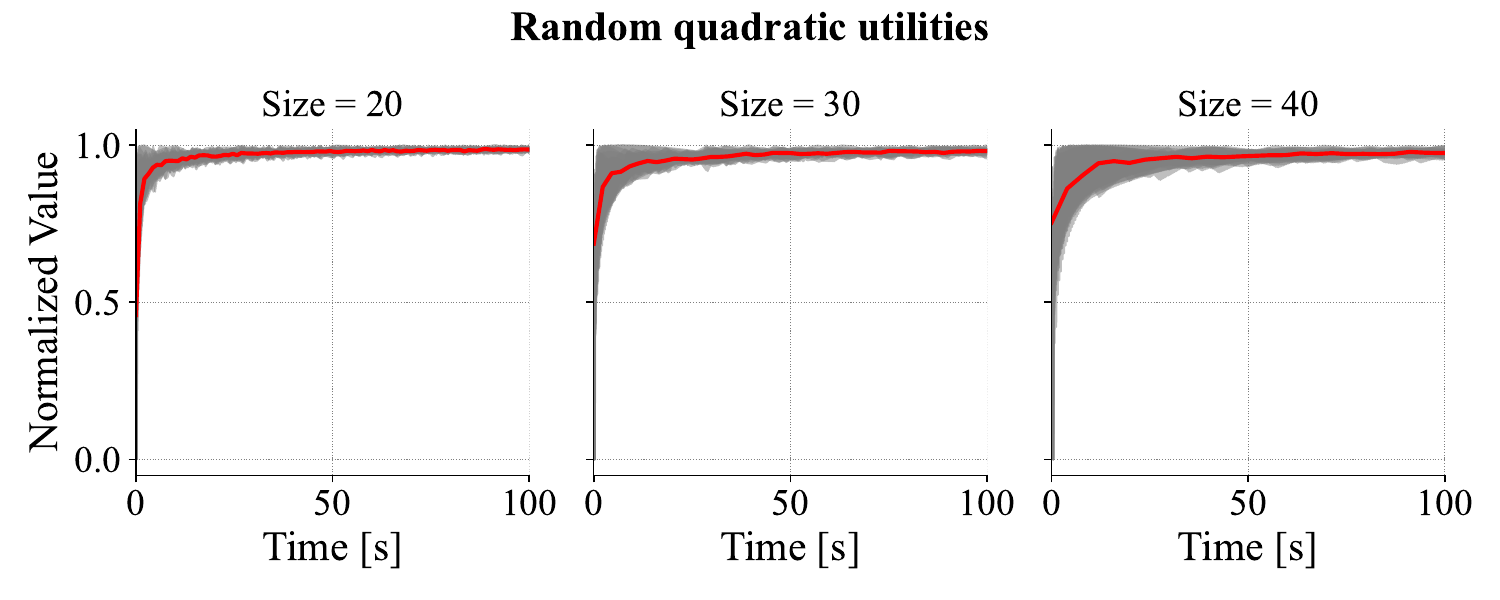}
\caption{Average (red) normalized objective value over 10 instances (gray) of a one-sided continuous two-level Blotto game with min aggregator and randomly generated quadratic battlefield utilities.}
\label{fig:quadr-indiv-values}
\end{figure}

\subsubsection{Security-inspired battlefield utilities}
We experiment on a real-world scenario inspired by security applications. We consider a one-sided continuous two-level Blotto game where Player 2 has 10 soldiers that she allocates over 3 battlefields. Each battlefield features a two‐player zero-sum security subgame in which Player 2 plays the defender’s role and Player 1 the attacker’s. For every battlefield $i$, we employ a payoff matrix $A_i$ from~\cite{krever2025guardconstructingrealistictwoplayer}, which generates realistic security‐game instances by leveraging diverse datasets, ranging from animal movement for anti‐poaching to demographic and infrastructure data for urban protection. This framework supplies both a suite of preconfigured instances and the ability to customize utility functions and game parameters. For our experiments, we select three urban‐protection scenarios set in Chinatown: Battlefield 1 models the northeast section over seven decision timesteps; Battlefield 2 uses the same region with eight timesteps; and Battlefield 3 covers the northwest section, also over eight timesteps. To incorporate Player 2’s continuous soldier allocation $\contsoldiers^2$ into the subgame payoffs, we set $u_i(\contsoldiers^2) = A_i + C_i\log(\contsoldiers+1)$ where $C_i$ is a positive normalization matrix used in subgame $i$ with entries $[C_i]_{ij}\sim\mathrm{Uniform}[0,1]$, and the logarithm is taken element-wise.
We executed~\Cref{alg:psa} for 1000 iterations using the same diminishing step-size as in the previous experiment, and tracking the objective value at each step. 
In the rightmost plot of \Cref{fig:conv_subgr_asc}, the convergence curve demonstrates the algorithm’s robustness in this real-world scenario.


\end{document}